\documentclass[pra,10pt,aps,notitlepage,twocolumn,nofootinbib,superscriptaddress,longbibliography]{revtex4-1}
\usepackage{algorithm,algorithmic,amsfonts,amsmath,amssymb,mathtools,mathrsfs,bbm,setspace}
\usepackage{graphicx,epic,eepic,epsfig,latexsym,verbatim,color}
\usepackage[shortlabels]{enumitem}

\usepackage{tikz}
\usepackage{mathrsfs}
\usetikzlibrary{chains}
\usetikzlibrary{fit}
\usetikzlibrary{arrows}
\usetikzlibrary{decorations}

\usepackage{epsfig}
\usetikzlibrary{shapes.symbols,patterns} 
\usepackage{pgfplots}

\usepackage[strict]{changepage}
\usepackage{hyperref}
\hypersetup{colorlinks=true,citecolor=blue,linkcolor=blue,filecolor=blue,urlcolor=blue,breaklinks=true}

\usepackage[marginal]{footmisc}
\usepackage{url}
\usepackage{theorem}

\newtheorem{theorem}{Theorem}
\newtheorem{lemma}[theorem]{Lemma}
\newtheorem{corollary}[theorem]{Corollary}
\newtheorem{proposition}{Proposition}
\newtheorem{definition}{Definition}

\newtheorem{remark}{Remark}

\def\squareforqed{\hbox{\rlap{$\sqcap$}$\sqcup$}}
\def\qed{\ifmmode\squareforqed\else{\unskip\nobreak\hfil
\penalty50\hskip1em\null\nobreak\hfil\squareforqed
\parfillskip=0pt\finalhyphendemerits=0\endgraf}\fi}
\def\endenv{\ifmmode\;\else{\unskip\nobreak\hfil
\penalty50\hskip1em\null\nobreak\hfil\;
\parfillskip=0pt\finalhyphendemerits=0\endgraf}\fi}
\newenvironment{proof}{\noindent \textbf{{Proof~} }}{\hfill $\blacksquare$}

\mathchardef\ordinarycolon\mathcode`\:
\mathcode`\:=\string"8000
\def\vcentcolon{\mathrel{\mathop\ordinarycolon}}
\begingroup \catcode`\:=\active
  \lowercase{\endgroup
  \let :\vcentcolon
  }

\usepackage{cleveref}
\usepackage{graphicx}
\usepackage{xcolor}
\usepackage{framed}
\definecolor{shadecolor}{rgb}{0.9,0.9,0.9}

\definecolor{darkblue}{RGB}{0,76,156}
\definecolor{darkkblue}{RGB}{0,0,153}
\definecolor{blue2}{RGB}{102,178,255}
\definecolor{darkred}{RGB}{195,0,0}

\RequirePackage[framemethod=default]{mdframed}
\newmdenv[skipabove=7pt,
skipbelow=7pt,
backgroundcolor=darkblue!15,
innerleftmargin=5pt,
innerrightmargin=5pt,
innertopmargin=5pt,
leftmargin=0cm,
rightmargin=0cm,
innerbottommargin=5pt,
linewidth=1pt]{tBox}

\newmdenv[skipabove=7pt,
skipbelow=7pt,
backgroundcolor=blue2!25,
innerleftmargin=5pt,
innerrightmargin=5pt,
innertopmargin=5pt,
leftmargin=0cm,
rightmargin=0cm,
innerbottommargin=5pt,
linewidth=1pt]{dBox}

\newmdenv[skipabove=7pt,
skipbelow=7pt,
backgroundcolor=darkred!15,
innerleftmargin=5pt,
innerrightmargin=5pt,
innertopmargin=5pt,
leftmargin=0cm,
rightmargin=0cm,
innerbottommargin=5pt,
linewidth=1pt]{rBox}

\newcommand{\nc}{\newcommand}
\nc{\bra}[1]{\langle#1|}
\nc{\ket}[1]{|#1\rangle}
\nc{\ketbra}[2]{|#1\rangle\!\langle#2|}
\nc{\braket}[2]{\langle#1|#2\rangle}

\DeclarePairedDelimiter{\abs}{\lvert}{\rvert}

\nc{\proj}[1]{| #1\rangle\!\langle #1 |}
\nc{\avg}[1]{\langle#1\rangle}
\nc{\rank}{\operatorname{Rank}}
\nc{\smfrac}[2]{\mbox{$\frac{#1}{#2}$}}
\nc{\tr}{\operatorname{Tr}}
\nc{\ox}{\otimes}
\nc{\dg}{\dagger}
\nc{\dn}{\downarrow}
\nc{\cA}{{\cal A}}
\nc{\cB}{{\cal B}}
\nc{\cC}{{\cal C}}
\nc{\cD}{{\cal D}}
\nc{\cE}{{\cal E}}
\nc{\cF}{{\cal F}}
\nc{\cG}{{\cal G}}
\nc{\cH}{{\cal H}}
\nc{\cI}{{\cal I}}
\nc{\cJ}{{\cal J}}
\nc{\cK}{{\cal K}}
\nc{\cL}{{\cal L}}
\nc{\cM}{{\cal M}}
\nc{\cN}{{\cal N}}
\nc{\cO}{{\cal O}}
\nc{\cP}{{\cal P}}
\nc{\cQ}{{\cal Q}}
\nc{\cR}{{\cal R}}
\nc{\cS}{{\cal S}}
\nc{\cT}{{\cal T}}
\nc{\cU}{{\cal U}}
\nc{\cV}{{\cal V}}
\nc{\cX}{{\cal X}}
\nc{\cY}{{\cal Y}}
\nc{\cZ}{{\cal Z}}
\nc{\cW}{{\cal W}}
\nc{\csupp}{{\operatorname{csupp}}}
\nc{\qsupp}{{\operatorname{qsupp}}}
\nc{\var}{{\operatorname{var}}}
\nc{\rar}{\rightarrow}
\nc{\lrar}{\longrightarrow}
\nc{\polylog}{{\operatorname{polylog}}}
\nc{\wt}{{\operatorname{wt}}}
\nc{\supp}{{\operatorname{supp}}}

\def\a{\alpha}

\def\e{\epsilon}

\def\x{\xi}

\nc{\RR}{{{\mathbb R}}}
\nc{\CC}{{{\mathbb C}}}
\nc{\FF}{{{\mathbb F}}}
\nc{\NN}{{{\mathbb N}}}
\nc{\ZZ}{{{\mathbb Z}}}
\nc{\PP}{{{\mathbb P}}}
\nc{\QQ}{{{\mathbb Q}}}
\nc{\UU}{{{\mathbb U}}}
\nc{\EE}{{{\mathbb E}}}
\nc{\id}{{\operatorname{id}}}

\nc{\CHSH}{{\operatorname{CHSH}}}

\newcommand{\Op}{\operatorname}

\nc{\rU}{\mbox{U}}

\nc{\ob}[1]{#1}

\nc{\SEP}{{\text{\rm SEP}}}
\nc{\NS}{{\text{\rm NS}}}
\nc{\LOCC}{{\text{\rm LOCC}}}
\nc{\PPT}{{\text{\rm PPT}}}
\nc{\EXT}{{\text{\rm EXT}}}
\nc{\Sym}{{\operatorname{Sym}}}

\nc{\ERLO}{{E_{\text{r,LO}}}}
\nc{\ERLOCC}{{E_{\text{r,LOCC}}}}
\nc{\ERPPT}{{E_{\text{r,PPT}}}}
\nc{\ERLOCCinfty}{{E^{\infty}_{\text{r,LOCC}}}}
\nc{\Aram}{{\operatorname{\sf A}}}

\makeatletter
\def\grd@save@target#1{%
  \def\grd@target{#1}}
\def\grd@save@start#1{%
  \def\grd@start{#1}}
\tikzset{
  grid with coordinates/.style={
    to path={%
      \pgfextra{%
        \edef\grd@@target{(\tikztotarget)}%
        \tikz@scan@one@point\grd@save@target\grd@@target\relax
        \edef\grd@@start{(\tikztostart)}%
        \tikz@scan@one@point\grd@save@start\grd@@start\relax
        \draw[minor help lines,magenta] (\tikztostart) grid (\tikztotarget);
        \draw[major help lines] (\tikztostart) grid (\tikztotarget);
        \grd@start
        \pgfmathsetmacro{\grd@xa}{\the\pgf@x/1cm}
        \pgfmathsetmacro{\grd@ya}{\the\pgf@y/1cm}
        \grd@target
        \pgfmathsetmacro{\grd@xb}{\the\pgf@x/1cm}
        \pgfmathsetmacro{\grd@yb}{\the\pgf@y/1cm}
        \pgfmathsetmacro{\grd@xc}{\grd@xa + \pgfkeysvalueof{/tikz/grid with coordinates/major step}}
        \pgfmathsetmacro{\grd@yc}{\grd@ya + \pgfkeysvalueof{/tikz/grid with coordinates/major step}}
        \foreach \x in {\grd@xa,\grd@xc,...,\grd@xb}
        \node[anchor=north] at (\x,\grd@ya) {\pgfmathprintnumber{\x}};
        \foreach \y in {\grd@ya,\grd@yc,...,\grd@yb}
        \node[anchor=east] at (\grd@xa,\y) {\pgfmathprintnumber{\y}};
      }
    }
  },
  minor help lines/.style={
    help lines,
    step=\pgfkeysvalueof{/tikz/grid with coordinates/minor step}
  },
  major help lines/.style={
    help lines,
    line width=\pgfkeysvalueof{/tikz/grid with coordinates/major line width},
    step=\pgfkeysvalueof{/tikz/grid with coordinates/major step}
  },
  grid with coordinates/.cd,
  minor step/.initial=.2,
  major step/.initial=1,
  major line width/.initial=2pt,
}
\makeatother

\usepackage{thmtools}
\usepackage{thm-restate}
\usepackage{etoolbox}
\makeatletter
\def\problem@s{}
\newcounter{problems@cnt}

\newcommand{\allproblems}{\problem@s}
\makeatother
\usepackage{amsmath,amsfonts,amssymb,array,graphicx,mathtools,multirow,bm,relsize,booktabs}
\usepackage{tabularx}

\usepackage{array}
\usepackage{xcolor}
\usepackage{inputenc}
\usepackage{bm}

\pgfplotsset{compat=1.18}

\usepackage{tikz-cd}
\tikzcdset{every label/.append style = {font = \small}}
\usepackage{tikz}
\usetikzlibrary{positioning}
\usetikzlibrary{shapes.geometric}
\usetikzlibrary{calc}
\definecolor{tensorblue}{rgb}{0.8,0.9,1}
\tikzset{ten/.style={fill=tensorblue}}

\allowdisplaybreaks
\begin{document}
\title{Circuit Knitting Faces Exponential Sampling Overhead Scaling Bounded by Entanglement Cost}
\author{Mingrui Jing}
\affiliation{Thrust of Artificial Intelligence, Information Hub, The Hong Kong University of Science and Technology (Guangzhou), Guangzhou 511453, China}
\author{Chengkai Zhu}
\author{Xin Wang}
\email{felixxinwang@hkust-gz.edu.cn}
\affiliation{Thrust of Artificial Intelligence, Information Hub, The Hong Kong University of Science and Technology (Guangzhou), Guangzhou 511453, China}

\begin{abstract}
Circuit knitting, a method for connecting quantum circuits across multiple processors to simulate nonlocal quantum operations, is a promising approach for distributed quantum computing. While various techniques have been developed for circuit knitting, we uncover fundamental limitations to the scalability of this technology. We prove that the sampling overhead of circuit knitting is exponentially lower bounded by the exact entanglement cost of the target bipartite dynamic, even with the asymptotic amount of resources using the parallel cut strategy. Specifically, we prove that the regularized sampling overhead assisted with local operations and classical communication (LOCC), of any bipartite quantum channel is lower bounded by the exponential of its exact entanglement cost under separable preserving operations. Furthermore, we develop the faithful lower bounds for the regularized sampling overhead based on channels' $\kappa$-entanglement and max-Rains information, providing efficiently computable benchmarks. Our work reveals a profound connection between virtual quantum information processing via quasi-probability decomposition and quantum Shannon theory, highlighting the critical role of entanglement in distributed quantum computing.
\end{abstract}

\maketitle


\section{Introduction}
Quantum computing, particularly through the lens of Distributed Quantum Computing (DQC)~\cite{cirac1999distributed}, is on the cusp of fundamentally altering the landscape of information processing by harnessing the unique properties of quantum mechanics. DQC aims to synergize multiple quantum processors into a cohesive, more potent system, addressing the inherent limitations in physical connectivity~\cite{van2016path} and laying the groundwork for advanced quantum networks and algorithms~\cite{parekh2021quantum,cacciapuoti2019quantum}. This endeavour is pivotal in the Noisy Intermediate-Scale Quantum (NISQ)~\cite{preskill2018quantum}, or even in the Early fault-tolerant quantum computing (EFTQC)~\cite{Katabarwa2024early} era, promising significant advancements in scalable quantum computation.

A notable contribution towards realizing near-term DQC is the concept of \textit{circuit knitting}~\cite{Bravyi2016,Piveteau2022circuit}. This approach not only facilitates the simulation of quantum many-body dynamics~\cite{Peng2020,Eddins2022,gentinetta2023overhead} but is also central in employing the quasi-probability decomposition (QPD) for mimicking bipartite quantum channels~\cite{Mitarai2021constructing,Mitarai2021overhead,Piveteau2022quasiprobability}. By operating a linear combination of local operation and classical communication (LOCC) channels separately on different quantum systems (FIG.~\ref{fig:limitation_knitting}), circuit knitting approximates the output of an intended quantum channel, suggesting a modular strategy for constructing expansive quantum circuits. This technique has spurred extensive research into quantum error mitigation~\cite{Temme2017,Endo2018,Kandala2019,Piveteau2021EM,singh2023experimental}, simulation algorithms~\cite{Pashayan2015,Seddon2021,gentinetta2023overhead}, and quantum resource theory~\cite{Howard2017,Seddon2019quantifying,Heinrich2019robustness}, highlighting its utility and innovative power.

Despite its potential, circuit knitting confronts substantial scalability and implementation challenges, primarily due to the prohibitive sampling costs associated with the QPD technique~\cite{Piveteau2022circuit}. These costs, which can exponentially increase with the circuit’s complexity, pose significant hurdles to the method's scalability, calling into question its feasibility for large-scale applications. Moreover, the precise factors influencing these costs remain poorly understood, complicating efforts to refine the technique.

Entanglement is a cornerstone resource in quantum computation and information, vital for processes like quantum teleportation~\cite{Gottesman1999demonstrating}, accelerated quantum algorithms~\cite{Ekert1998quantum}, and quantum error correction~\cite{Bennett1996}. The need for efficient entanglement utilization has led to the development of the exact (parallel) entanglement cost concept for bipartite quantum channels, which quantifies the minimal rates of maximally entangled states required for simulating multiple uses of error-free operations~\cite{Bauml2019resource,Gour2021entanglement,Gour2020}. This measure has significantly impacted studies in entanglement manipulation, channel resource theory, and information amortization~\cite{Bauml2019resource,Gour2021entanglement,Kaur2017amortized}, prompting us to inquiry into its relationship with the sampling overhead in circuit knitting.

Originally introduced to set upper bounds on channel capacities~\cite{Bennett1996}, entanglement cost has been extensively analyzed, especially regarding its ties to entanglement of formation and channel simulation~\cite{Berta2013,Wilde2018}. Despite its theoretical importance in channel simulation, the practical significance of entanglement cost in quantum computing and many-body physics remains less clear. Challenges in computing entanglement cost for general bipartite channels arise from complexities in superchannels, causal ordering, and regularization~\cite{Chiribella2008quantum,Bauml2019resource}. However, recent advancements in computable bounds have spurred further research in entanglement theory and quantum resource theories, enhancing our understanding of quantum information processes~\cite{Wang2017irreversibility,Gour2021entanglement,Wang2023,Lami2023no}.

\subsection{Summary of results}
In this work, we try to establish the quantitative relationship between the two ways of channel simulations (Fig~\ref{fig:limitation_knitting}). We reveal the fundamental limitations of circuit knitting by theoretically establishing an exponential lower bound of the sampling cost with respect to the exact entanglement cost of the bipartite channel to be knitted. Technically, we further design efficiently computable lower bounds for the zero-error asymptotic sampling cost of any bipartite channels,
which can be efficiently estimated via semidefinite programming (SDP). Particularly, we show the following 
\begin{equation}
\begin{aligned}
    & \gamma_{\LOCC}^{\infty}(\cN_{AB\rightarrow A'B'}) \geq 2^{E^{\SEP}_{C,0}(\cN_{AB\rightarrow A'B'})},\\
    & \gamma_{\PPT}^{\infty}(\cN_{AB\rightarrow A'B'}) \geq 2^{E^{\PPT}_{C,0}(\cN_{AB\rightarrow A'B'})},
\end{aligned}
\end{equation}
where $\gamma_{\LOCC}^{\infty}(\cN), \gamma_{\PPT}^{\infty}(\cN)$ are LOCC-assisted and PPT-assisted regularized sampling cost of $\cN_{AB\rightarrow A'B'}$, respectively. $E^{\SEP}_{C,0}(\cN_{AB\rightarrow A'B'})$ and $E^{\PPT}_{C,0}(\cN_{AB\rightarrow A'B'})$ are the exact entanglement cost of $\cN_{AB\rightarrow A'B'}$ via PPT bipartite operations and \textit{separable} (SEP) bipartite operations, respectively~\cite{Gour2021entanglement,Bauml2019resource}. Notably, our results apply to general bipartite channels instead of specific bipartite unitaries and break through the barrier of \textit{KAK} decomposition~\cite{Tucci2005introduction} mainly used in previous literature~\cite{Mitarai2021constructing,Piveteau2022circuit}. We continue to study noisy two- and three-qubit gates, such as CNOT, Toffoli, and control SWAP gates, and showcase the different costs via different partitioning.

\begin{figure}[t]
    \centering
    \includegraphics[width=\linewidth]{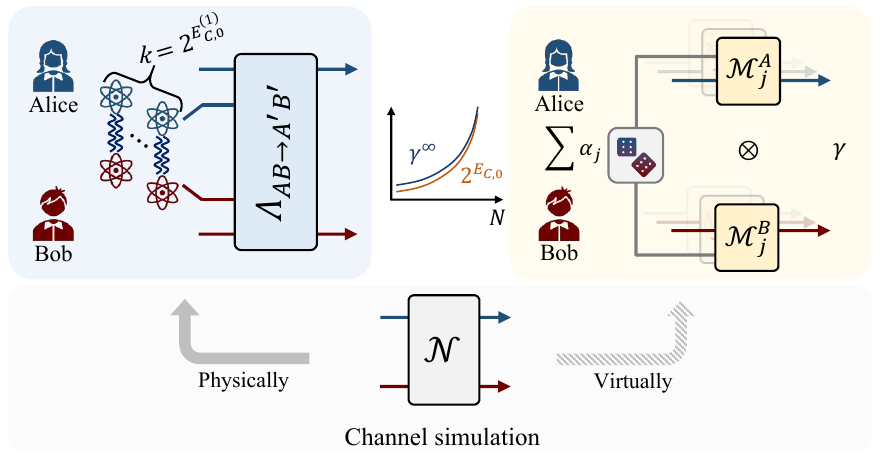}
    \caption{\textbf{Illustration of the main result: the relationship between the exact (parallel) entanglement cost $E_{C,0}$ of implementing a bipartite quantum channel $\cN$ (left) and the regularized sampling overhead $\gamma^{\infty}$ for simulating $\cN$ via circuit knitting (right)}. The one-shot exact entanglement cost $E_{C,0}^{(1)}$ measures the minimum entanglement required for Alice and Bob to exactly realize $\cN$. The implementation uses a bipartite channel $\Lambda_{AB \rightarrow A'B'}$ from LOCC, SEP, or PPT operations. The $\gamma$-factor indicates the minimal overhead for simulating $\cN$ with QPD. Our results highlight the limitations of the circuit knitting technique, establishing an exponential lower bound of $\gamma^{\infty}$ relative to $E_{C,0}$ for simulating numerous instances of $\cN$.}
    \label{fig:limitation_knitting}
\end{figure}

Our studies give a deeper understanding of the scalability problem for circuit knitting, indicating that, even with shared bound entanglement~\cite{Horodecki1998mixed,Horodecki1999bound,Gaida2023seven}, the sampling cost might still not be reduced from the exponential-growing catastrophe~\cite{Piveteau2022circuit}. Despite the fundamental limitations of circuit knitting from this work in terms of quantum Shannon theory, our discussion laterally implies the potential quantum advantages of entangling channels, which might not be simply replaced by the sampling simulation techniques.

This paper is structured as follows. Sec.~\ref{sec:preliminary} provides preliminaries about quantum channel simulation, introducing two perspectives of bipartite quantum channel simulation, including the QPD method and the conventional method using LOCC and consuming entanglement resources. Sec.~\ref{sec:sampling_cost_and_bipartite_channels} exhibits the main results from our investigation, including the exponential lower bounds of regularized $\gamma$-factor regarding the entanglement cost of bipartite channels and computable quantum information measures. In Sec.~\ref{sec:fundamental_lim_circuit_knitting}, the arguments on the fundamental limitation of circuit knitting will be given provided the PPT-assisted QPD and our lower bound results. The paper concludes with a summary and outlooks for future research in Sec.~\ref{sec:conclusion}.

\section{Preliminaries}\label{sec:preliminary}

\subsection{Simulation of bipartite quantum channels}\label{subsec:simulation_bipartite_qc}
The simulation of a bipartite quantum channel can be understood from two perspectives. The first one virtually simulates the channel by sampling quantum operations and post-processing, which can be used to estimate the expectation value of any physical observable after the state passes through the channel. The second one utilizes some entangled states between the two parties to physically realize the bipartite channel via entanglement non-generating operations. 

\paragraph{Simulation via Monte-Carlo method.}
Let us begin our discussion on the channel simulation task by sampling and post-processing. Given any bipartite quantum channel $\cN_{AB\rightarrow A'B'}$, we target to find a set of real coefficients $\{\alpha_j\}$ and a set of LOCC $\{\cM_j\}$ between $\cL(\cH_{AA'})$ and $\cL(\cH_{BB'})$ such that $\cN = \sum_j \alpha_j \cM_j$. The set $\{(\alpha_j, \cM_j)\}$ forms a quasiprobability decomposition (QPD) of the channel $\cN$~\cite{Mitarai2021constructing,Mitarai2021overhead,Piveteau2022quasiprobability}. Such a QPD protocol of $\cN$ enables the simulation method of $\cN$ via randomly sampling $\cM_j$'s with their regarding probability $|\alpha_j|/\kappa$ where $\kappa = \sum_j \abs{\a_j}$. That is, for any (normalized) observable $O$ and input quantum state $\rho$, the expectation value of $O$ with respect to $\cN(\rho)$ can be extracted by the Monte Carlo method. One can also extend the region of the decomposition channels to \textit{trace non-increasing} maps~\cite{Piveteau2022circuit}, which may yield extra samples during physical implementations for the ancillary post-processing.

By \textit{Hoeffding's inequality}, we require at least $\cO(\frac{\kappa^2}{\delta^2}\ln{\frac{2}{\epsilon}})$ samples of $\cM_j$ to achieve the estimation error $\delta$ with a probability $1-\epsilon$. Therefore, the sampling overhead defined as $\kappa$ determines the number of samples demanded to simulate $\cN$. Among all possible LOCC-assisted QPD of $\cN$, we could quantify the sampling cost of simulating $\cN$ by identifying the minimum sampling overhead, denoted as $\gamma_{\LOCC}(\cN)$ which also quantifies the nonlocality of given bipartite quantum channel~\cite{Piveteau2022circuit}. 

Notice that the $\gamma$-factor is closely related to the robustness of the quantum channel regarding specific channel sets and has been discussed in Ref.~\cite{Harrow2024optimal,Kim2021one}. The definition of QPD for circuit knitting concerning other operation sets has been discussed in~\cite{Piveteau2022circuit,Schmitt2023cutting}. The technique also inspires the `twisted channel' framework for general measurement in Ref.~\cite{Zhao2023power}.

\paragraph{Bipartite channel simulation via LOCC and entanglement.}
Compared to the QPD sampling simulation, the entanglement cost of a point-to-point quantum channel, i.e., the channel with single input and output systems, is defined as the minimal asymptotic rate of the maximally entangled state required to simulate $n$-use of the channel assisted with LOCC~\cite{Berta2013}. The framework was then extended to the more complex bipartite scenarios based on the channel simulation tasks~\cite{Bauml2019resource,Gour2020,Gour2021entanglement} embedded in the framework of quantum comb~\cite{Chiribella2008quantum}. To distinguish with the point-to-point cases, the entanglement cost of bipartite channels is more complicated, containing both sequential and parallel schemes. The former allows the amortization of information to enhance the simulation, while the other naturally forms a particular case of the sequential setting by consuming independent resources. In Refs.~\cite{Bauml2019resource,Gour2021entanglement}, both independent works have developed efficiently computable lower bounds of the costs via the relaxation towards PPT superchannels. 

Given a bipartite quantum channel $\cN_{AB\rightarrow A'B'}$, one can define the entanglement cost of $\cN$ via the task of channel simulations~\cite{Berta2013,Bauml2019resource,Wilde2018,Gour2021entanglement}. Two types of schemes, including the parallel and the adaptive simulation scheme, are used to define the entanglement cost of bipartite channels. Under the parallel framework, the goal is to simulate $\cN^{\ox n}$ simultaneously using some higher-dimensional free channels consuming maximally entangled states. Let us consider a simulation protocol $\cP_{A^n B^n \rightarrow A'^n B'^n}$ such that for any input state $\rho_{A^n B^n}$, we have,
\begin{equation}\label{eq:protocol_of_bipartite_simulation}
    \cP(\rho_{A^n B^n})\coloneqq \Lambda_{A^n B^n \Bar{A}\Bar{B} \rightarrow A'^n B'^n}(\rho_{A^n B^n} \ox \Phi_{\Bar{A}\Bar{B}}),
\end{equation}
where $\Lambda$ here is a free channel, e.g., LOCC-, PPT- or SEP-preserving channels. For $\varepsilon \geq 0$, we say the protocol is $\varepsilon$-distinguishable from $\cN^{\ox n}$ if,
\begin{equation}\label{eq:simulation_protocol_diamond_distance}
    \|\cN_{AB\rightarrow A'B'}^{\ox n} - \cP_{A^n B^n \rightarrow A'^n B'^n}\|_{\diamond} \leq \varepsilon,
\end{equation}
where $\|\cdot\|_{\diamond}$ is the diamond norm between linear maps~\cite{Kitaev1997QuantumCA}. Based on the channel simulation protocol, we can first define the one-shot entanglement cost of a bipartite channel,

\begin{definition}
    Consider a bipartite system with parties Alice and Bob. Let $\varepsilon \geq 0$ and $\cN_{AB\rightarrow A'B'}$ be a bipartite channel. The one-shot entanglement cost of $\cN$ with error $\varepsilon$, regarding the free operation set $\cF$, is defined as
    \begin{equation}
    \begin{aligned}
        E_{\cF, C,\varepsilon}^{(1)}(\cN) := &\min\big\{\log k \, :\, \big\|\cN_{AB\rightarrow A'B'} - \\
        &\Lambda_{A B \Bar{A}\Bar{B} \rightarrow A' B'}( \ \cdot \ox \Phi_{\Bar{A}\Bar{B}}^{\ox k})\big\|_{\diamond}\leq \varepsilon, \, k\in \mathbb{N} \big\},
    \end{aligned}
    \end{equation}
    where the minimization ranges over all $\Lambda$ operations in $\cF$ between Alice and Bob.
\end{definition}

Therefore, the (parallel) $\cF$-entanglement cost of bipartite channels can be defined as the asymptotic infimum rate of the entanglement resources to realize the protocol with zero error, i.e.,
\begin{equation}\label{eq:asymptotic_cost_bipartite_channel}
    E^{\cF}_C(\cN) \coloneqq \lim_{\varepsilon \rightarrow 0} \limsup_{n\rightarrow \infty} \frac{1}{n}E_{\cF, C, \varepsilon}^{(1)}(\cN^{\ox n}).
\end{equation}
We can also define the exact (parallel) $\cF$-entanglement cost of $\cN$ as,
\begin{equation}\label{eq:exact_cost_bipartite_channel}
    E^{\cF}_{C,0}(\cN) \coloneqq \limsup_{n\rightarrow \infty} \frac{1}{n}E_{\cF, C, \varepsilon=0}^{(1)}(\cN^{\ox n}),
\end{equation}
which forms an upper bound of the asymptotic entanglement cost. In particular, if $\cF = \LOCC$, we could omit to label `$\cF$' in the symbol of entanglement cost. Refer to the simulation of channels, in the single-shot scenario, one requires $\Phi_{\Bar{A}\Bar{B}} = \Phi_{\Bar{A}_0\Bar{B}_1} \ox \Phi_{ \Bar{A}_1\Bar{B}_0}$ and the free operation to simulate the SWAP operation $\cS$ where $\Bar{A} = \Bar{A}_0\ox\Bar{A}_1$ and $\Bar{B} = \Bar{B}_0\ox\Bar{B}_1$. In the parallel setting, one can, {therefore,} treat the SWAP operation as the maximal resource for dynamical entanglement~\cite{Bauml2019resource,Gour2020}. Notice that, taking {the dimension} $|B| = |A'| = 1$, the above definition is naturally reduced to the entanglement cost of the point-to-point channel. Another definition of the entanglement cost of bipartite channels, namely the sequential or adaptive cost, has been also generally discussed in Refs.~\cite{Bauml2019resource,Gour2021entanglement}, forming a lower bound of the parallel entanglement cost.

\section{Sampling cost is exponentially lower bounded by entanglement cost
}\label{sec:sampling_cost_and_bipartite_channels}

Despite many efforts that have been made recently, the computation of the $\gamma$-factor for knitting arbitrary bipartite quantum channels via LOCC stays challenging. Besides, the relationship between the sampling cost and the entanglement theory of bipartite channels has not been completely built. {The $\gamma$-factor was proven to be captured by the robustness of entanglement~\cite{Piveteau2022circuit}, which is initially defined by Vidal and Tarrach~\cite{vidal1999robustness} in the static entanglement scenario, and further extended to general resources~\cite{Takagi2019general}. It serves as a bridge between both the limitation of entanglement distillation~\cite{regula2021fundamental} and the entanglement manipulation through finite calls~\cite{Kim2021one}, and, recently, inspires computable lower bounds of dynamical entanglement cost~\cite{Lami2023computable}. In particular, we have,
\begin{equation}\label{eq:relation_gamma_Rob}
    \gamma_{\cF}(\cN_{AB\rightarrow A'B'}) = 1+2R^{\cF}_s(\cN_{AB\rightarrow A'B'})
\end{equation}
where $R^{\cF}_s(\cdot)$ is the \textit{standard robustness of a quantum channel based on the free operation set $\cF$}.} In this section, we first introduce the settings of circuit knitting and the definitions of the exact regularized $\gamma$-factor from the point of resource theory. In subsection~\ref{subsec:ppt_exact_entanglement_cost_lb_for_sampling_cost}, we introduce the PPT-assisted $\gamma$-factor and derive the exponential lower bound of its regularized form with respect to the non-positive-partial-transpose (NPT) exact entanglement cost, and provide an efficiently computable upper bound of it via SDP. In subsection~\ref{subsec:sep_exact_entanglement_cost_lb_for_sampling_cost}, we further tighten the two sides of the theorem towards the LOCC scenario by proving the similar relationship for separable channels.

With the optimal QPD of $\cN$, we can imagine a situation of a quantum circuit containing $n$ identical bipartite channels $\cN$ applied to the two parties of the quantum circuit which is used in many quantum computing tasks. For example, in the estimation of the ground state energy of the Hamiltonian via variational ansatz~\cite{Cerezo2021variational} and simulation of quantum many-body dynamics~\cite{Harrow2024optimal}. One can avoid generating global entanglement by individually decomposing those entangling operations across the two subcircuits. Unfortunately, this would result in an exponential growth in sampling overhead~\cite{Piveteau2022circuit}.

From the perspective of quantum resource theory, Ref.~\cite{Piveteau2022circuit} has introduced \textit{parallel cutting} for circuit knitting to reduce sampling cost effectively. Consider a circuit with $n$ multiple gates $(U_j)_{j=1}^n$, which we wish to simultaneously. In that case, the entire cut can be observed as a single cut on the total circuit $U_{tot} = \bigotimes_{j=1}^n U_j$ via some pre-processes using local swap operations. 
We can define the effective sampling overhead per gate of cutting~\cite{Piveteau2022circuit} as,
\begin{equation}\label{eq:effective_gamma_factor}
    \gamma^{(n)}_{\cF}(\cN_{AB\rightarrow A'B'}) \coloneqq (\gamma_{\cF_n}(\cN_{AB\rightarrow A'B'}^{\otimes n}))^{1/n},
\end{equation}
where $\cF_n$ is the free operation set of $n$-shot. The optimal strategy from joint cut indicates the limitation of cutting $U_{tot}$ into QPD and the total number of samples required can gain significant reduction from $\prod_{j=1}^n \gamma_{\cF}(\cU_j)^2$ to $\prod_{j=1}^n \gamma^{(n)}_{\cF}(\cU_j)^2$. Despite the subtleties that determining $U_{tot}$ itself is as challenging as implementing the entire circuit directly~\cite{Schmitt2023cutting}, we can still, in theory, define the regularized exact $\gamma$-factor with respect to the free operation set $\cF$ by taking the limit $n\rightarrow \infty$ for any bipartite channel $\cN$, i.e., $\gamma_{\cF}^{\infty}(\cN)\coloneqq \lim_{n\rightarrow \infty} \gamma_{\cF}^{(n)}(\cN)$, which quantifies the fundamental limitation on the sampling cost from simulating $\cN$ using parallel cutting.

\subsection{PPT-assisted sampling cost and NPT exact entanglement cost}\label{subsec:ppt_exact_entanglement_cost_lb_for_sampling_cost}
In the NPT resource theory, PPT channels are commonly applied in studies on the properties and limitations of LOCCs. Based on Refs.~\cite{Bauml2019resource,Gour2021entanglement,Gour2020}, one can define the exact PPT-entanglement cost of any bipartite channel $\cN_{AB\rightarrow A'B'}$ by letting $\cF = \PPT$ in Eq.~\eqref{eq:exact_cost_bipartite_channel} and it is known that $E^{\PPT}_{C,0}(\cN_{AB\rightarrow A'B'})$ is smaller than or equal to the max-logarithmic negativity or generalized $\kappa$-entanglement $LN_{\max}\left(\cN_{AB\rightarrow A'B'}\right)$, i.e., $E^{\PPT}_{C,0}(\cN_{AB\rightarrow A'B'}) \le LN_{\max }\left(\cN_{AB\rightarrow A'B'}\right)$, and 
\begin{equation}
\begin{aligned}
    LN_{\max }&\left(\cN\right)=\log\inf \Big\{\max \left\{\left\|P_{AB}\right\|_{\infty},\left\|P_{AB}^{T_{B}}\right\|_{\infty}\right\}:\\
    &-P_{ABA'B'}^{T_{BB'}} \leq\left(J_{ABA'B'}^{\cN}\right)^{T_{BB'}} \leq P_{ABA'B'}^{T_{BB'}}\Big\},
\end{aligned}
\end{equation}
where $P_{ABA'B'} \geq 0$ is some optimization variable. The max-logarithmic negativity of a bipartite channel can be computed via SDP~\cite{Gour2021entanglement}, which is a generalization of $\kappa$-entanglement~\cite{Wang2020c} for the bipartite channels.
It is called max-logarithmic negativity since it can be understood as the maximum case of the $\alpha$-logarithmic negativity~\cite{Wang2020alpha}.

Our first result is to establish the relationship between $\gamma_{\PPT}(\cN_{AB\rightarrow A'B'})$ and the exact PPT-entanglement cost of the channel.
\begin{theorem}
    For a bipartite quantum channel $\cN_{AB\rightarrow A'B'}$, the PPT-assisted regularized $\gamma$-factor is exponentially lower bounded by,
    \begin{equation}
        \gamma^{\infty}_{\PPT}(\cN_{AB\rightarrow A'B'}) \geq 2^{LN_{\max}(\cN_{AB\rightarrow A'B'})},
    \end{equation} 
    where $LN_{\max}(\cN) \geq E^{\PPT}_{C,0}(\cN)$, and $E^{\PPT}_{C,0}(\cN)$ denotes the exact (parallel) PPT-entanglement cost of $\cN$.
\end{theorem}
The detailed proof is in Appendix~\ref{appendix:proof_of_regularized_ppt_sampling_cost}. We remark that the last inequality is not an equality as presented in~\cite{Gour2020} due to the non-additivity of max-logarithmic negativity or $\kappa$-entanglement~\cite{LamiRegula23private,WW2024}. {Notice that the max-logarithmic negativity can be evaluated via SDP, hence, forming a computable lower bound of regularized PPT-assisted $\gamma$-factor.}

We can also derive another efficiently computable lower bound of $\gamma^{\infty}_{\PPT}(\cN)$ via the max-Rains information~\cite{Wang2019}, whose bidirectional version was applied to establish the upper bound for entanglement generating capacity of any bipartite channels~\cite{Bauml2018}. Our second bound reads,
\begin{equation}\label{eq:max_rains_info_lb_regularized_gamma_ppt}
    \gamma^{\infty}_{\PPT}(\cN_{AB\rightarrow A'B'}) \geq 2^{R_{\max}^{2\rightarrow 2}(\cN_{AB\rightarrow A'B'})}.
\end{equation}
The max-Rains information of a bipartite quantum channel $\cN_{AB\rightarrow A'B'}$ is defined as, 
\begin{equation}
    R_{\max}^{2\rightarrow 2}(\cN) \coloneqq \log(\Gamma^{2\rightarrow 2}(\cN)),
\end{equation}
where $\Gamma^{2\rightarrow 2}(\cN)$ can be evaluated via the SDP,
\begin{equation}
\begin{aligned}
    \Gamma^{2\rightarrow 2}(\cN) = \inf\big\{\| &\tr_{AB}\{V_{AA'BB'} + Y_{AA'BB'}\}\|_{\infty} : \\
    &(V - Y)^{T_{BB'}} \geq J^{\cN}_{AA'BB'}\big\},
\end{aligned}
\end{equation}
where $V,Y\geq 0$. Since our bounds work for general bipartite channels, we can then apply our bounds to investigate the variations of sampling cost for noisy gates. For example, consider a noisy CNOT gate $\cN_p = \cD_{p}\circ \cC\cX$ where $\cD_p$ is a two-qubit depolarizing channel. For any two-qubit state $\rho_{AB}$, we have,
\begin{equation}
    \cN_p(\rho_{AB})\coloneqq (1-p)\Op{CNOT}\rho_{AB}\Op{CNOT} + p \frac{I_{AB}}{d_{AB}}.
\end{equation}
As we can observe from Fig.~\ref{fig:maxRains_vs_maxLog}, when there is no noise occurring, $\gamma_{\PPT}(\cN_p)$ can reach $3$ from our SDP calculation. This is the same value as from the LOCC results of the CNOT gate derived in Ref.~\cite{Piveteau2022circuit}. The observation indicates no advantages of using bipartite PPT channels in decomposing noisy CNOT. The two lower bounds from $LN_{\max}$ and $R_{\max}^{2\rightarrow 2}$ both start at $2$, declaring at least one ebit is required to implement a CNOT with PPT or LOCC. When $p\geq 0.9$, too much noise destroys the entangling power of the CNOT gate, and all bounds reduce to one. 

An interesting phenomenon from Fig.~\ref{fig:maxRains_vs_maxLog} can be observed, that is, the two lower bounds coincide with each other. This does not hold for general cases. We provide the following counterexample to show the relative size of the two bounds. Suppose now a SWAP gate is applied where the second qubit is accidentally stroked by the qubit amplitude damping (AD) channel $\cA_{\gamma}$ defined by its Kraus operators 
\begin{equation}
    K_0 = \ketbra{0}{0} + \sqrt{1-\gamma}\ketbra{1}{1},\, K_1 = \sqrt{\gamma}\ketbra{0}{1}.
\end{equation}
AD channel models the loss of quantum information when a quantum system interacts with its thermal bath. Our noisy SWAP operation is then defined as
\begin{equation}
    \cS_{\gamma}(\rho_{AB})\coloneqq (\cI_A \ox \cA_{\gamma})\left(\Op{SWAP} \rho_{AB}\Op{SWAP}\right)
\end{equation}
for any two-qubit state $\rho_{AB}\in \cD(\cH_{AB})$. The results have been shown in Fig.~\ref{fig:boundsADSWAP} where a distinct gap between the one-shot $\gamma_{\PPT}(\cS_{\gamma})$ and the two lower bounds can be identified. Notice that $LN_{\max}(\cS_{\gamma})$ does not necessarily equal to $R_{\max}^{2\rightarrow2}(\cS_{\gamma})$, where in this case, the max logarithmic negativity bound is slightly larger than the max-Rains information bound as the damping rate varies.

\begin{figure}[t]
    \centering
    \includegraphics[width=0.9\linewidth]{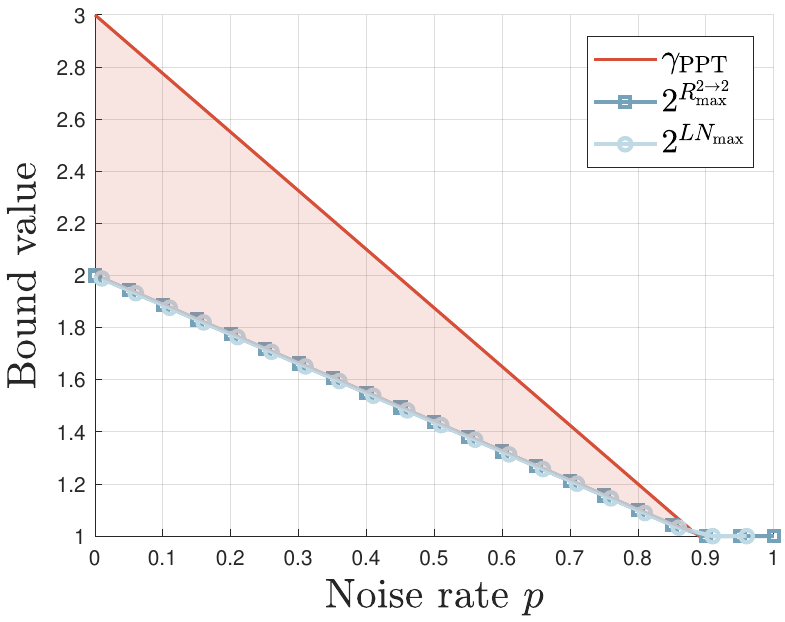}
    \caption{\textbf{Lower bounds comparison for noisy CNOT operation undertaking the two-qubit depolarizing channel by the noise rate $p$.} The orange solid line gives the one-shot exact $\gamma_{\PPT}(\cN)$. The dark and light blue dashed lines with square and circle markers illustrate the lower bounds of regularized PPT $\gamma$-factor by max-Rains information and the max logarithmic negativity, respectively. 
    {The shaded area is the possible region for $\gamma_{\PPT}^{\infty}(\cN)$ to lie in.}
    }
    \label{fig:maxRains_vs_maxLog}
\end{figure}

\begin{figure}[t]
    \centering
    \includegraphics[width=0.9\linewidth]{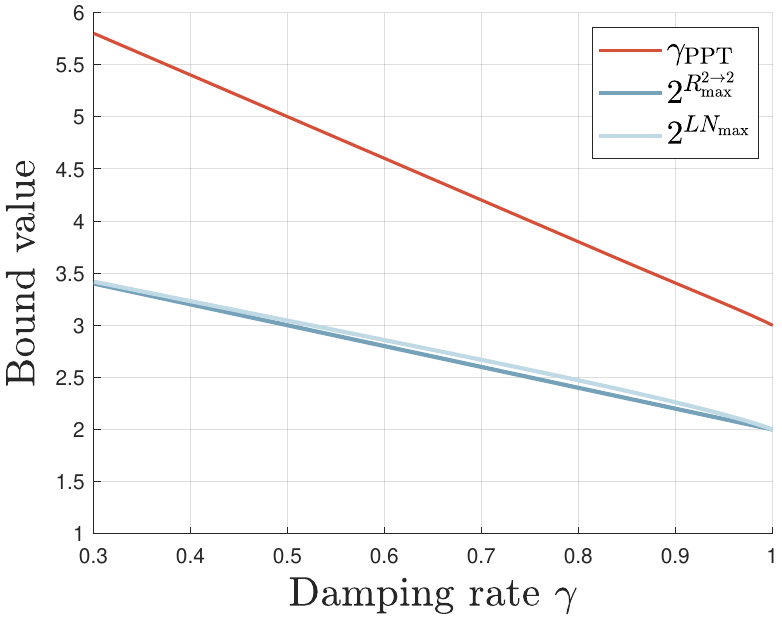}
    \caption{\textbf{Lower bounds comparison for noisy SWAP operation undertaking the single-qubit amplitude damping channel by the damping rate $\gamma$.} The orange solid line gives the one-shot exact $\gamma_{\PPT}(\cN)$. The dark and light blue lines illustrate the lower bounds by max-Rains information and the max logarithmic negativity, respectively.}
    \label{fig:boundsADSWAP}
\end{figure}

\subsection{SEP-assisted sampling cost and SEP exact entanglement cost}\label{subsec:sep_exact_entanglement_cost_lb_for_sampling_cost}
In the previous section, we have extended the set of LOCC operations to PPT operations for more theoretical convenience. Here, we concentrate on a smaller free operation set $\cF = \SEP$, namely the \textit{separable-preserving} channels defined in~\cite{Gour2021entanglement}. A bipartite channel is separable if and only if its Choi matrix is separable. Therefore, every separable quantum channel, at the same time, forms a PPT channel which implies that $\gamma_{\PPT}(\cN)\leq \gamma_{\SEP}(\cN)$. Besides, each above entanglement cost measure regarding SEP channels forms an upper bound of the corresponding PPT-entanglement cost. The set of SEP channels also serves as a candidate for investigating the properties of LOCC channels in quantum information theory, despite,  detecting the separability of any quantum states is NP-hard~\cite{Doherty2004complete}.

{Inspired by the entanglement manipulation for quantum channels~\cite{Kim2021one,Chitambar2017entanglement},} we can establish a tighter bound towards LOCC-assisted circuit knitting based on {~\cite[Theorem III.1.]{Kim2021one}. 
\begin{theorem}[\cite{Kim2021one}]
Given $\varepsilon \geq 0$, the one-shot dynamic entanglement cost of a bipartite quantum channel $\mathcal{N}_{A B}$ under seperability-preserving superchannel is bounded as
\begin{equation}
    L R_s^{\varepsilon}\left(\cN_{A B}\right) \leq E_{C,\varepsilon}^{(1)}\left(\cN_{A B}\right) \leq L R_s^{\varepsilon}\left(\cN_{A B}\right)+2.
\end{equation}
\end{theorem}
We further observe the relation between} the regularization forms of both quantities and establish the following corollary.
\begin{corollary}
    For a bipartite quantum channel $\cN_{AB\rightarrow A'B'}$, the SEP-assisted regularized $\gamma$-factor is exponentially lower bounded by the exact (parallel) SEP-entanglement cost of $\cN$,
    \begin{equation}
        \gamma^{\infty}_{\SEP}(\cN_{AB\rightarrow A'B'}) \geq 2^{E^{\SEP}_{C,0}(\cN_{AB\rightarrow A'B'})}.
    \end{equation}
\end{corollary}    

The detailed proof can be found in Appendix~\ref{appendix:proof_of_regularized_sep_sampling_cost}. Through the computable bound from PPT operations, one may allow `too much' freedom to make a coarse relaxation for LOCC entanglement. Our ultimate goal is to establish the LOCC sampling cost lower bound by the exponential LOCC entanglement cost of given bipartite channels. All separable operations form a much more restricted set {compared to PPT operations, and} contains all LOCC channels. Despite their unrealizability in some cases, to our best knowledge, the usage of SEP channels leads to much tighter bounds for both the sampling cost and the entanglement costs compared to the PPT results, i.e., $\gamma^{\infty}_{\SEP}(\cN) \leq \gamma^{\infty}_{\LOCC}(\cN)$ and $E^{\SEP}_C(\cN) \leq E_{C}(\cN)$. As a consequence, combining all of the above results, one can derive the following remark connecting both the LOCC circuit knitting sampling complexities and the entanglement cost measures of bipartite quantum channels,
\begin{remark}
    For any bipartite quantum channel $\cN_{AB\rightarrow A'B'}$, the following inequalities holds,
    \begin{equation}
    \begin{aligned}
        &\gamma^{\infty}_{\PPT}(\cN) \geq  2^{E_{C,0}^{\PPT}(\cN)} \geq 2^{E_{C}^{\PPT}(\cN)},\\ &\gamma^{\infty}_{\SEP}(\cN) \geq 2^{E_{C,0}^{\SEP}(\cN)} \geq 2^{E_{C}^{\SEP}(\cN)}.
    \end{aligned}
    \end{equation}
\end{remark}

We can also demonstrate that the LOCC regularized $\gamma$-factor is exponentially lower bounded by the entanglement cost of $\cN$ if $\cN$ forms a bicovariant channel~\cite{Bauml2019resource}, {for example, CNOT.} The proof of this can be found in Appendix~\ref{appendix:lb_for_smoothed_regularized_locc_via_choi_state}. However, whether or not the above relation holds in general stays open, and we leave this for future investigation.

\section{Fundamental limitation of circuit knitting}\label{sec:fundamental_lim_circuit_knitting}
In previous sections, we demonstrated the quantitative relation between the sampling cost of simulating a general bipartite channel via circuit knitting and the corresponding zero-error entanglement cost of the channel. Our theorem proves the intuition that the more entanglement resource cost for physically implementing the channel, the harder it is to simulate it via sampling LOCC channels combined with post-measurement selections. Based on the results, we attempt to establish the fundamental limitation of the circuit knitting technique in this section: in  subsection~\ref{subsec:sampling_cost_for_specific_operations}, we revisit the sampling cost for some important quantum logic gates via new bounds; in subsection~\ref{subsec:exponential_cost_for_knitting_multiple_instances}, we {investigate the exponential growth of sampling cost for simulating multiple instances of bipartite channels in a more practical scenario.}

\subsection{Sampling cost for specific operations}\label{subsec:sampling_cost_for_specific_operations}

The investigation of both the CNOT and SWAP gates stands at the central position of both quantum circuit architecture~\cite{Kielpinski2002architecture} and quantum information theory. A CNOT gate and a Bell state are commonly treated as equivalent resources of entanglement since a CNOT gate can be constructed via gate teleportation consuming one Bell state~\cite{Bennett1993}. A SWAP gate owns the maximal resource of dynamical entanglement~\cite{Gour2020}, which can generate two ebits shared with bi-parties. Besides, CNOT itself, together with single-qubit Pauli gates, forms a \textit{universal} set for realizing any logic gates~\cite{Nielsen2010quantum}, while the SWAP gate is closely related to the general two-qubit Clifford gates~\cite{Piveteau2022circuit}.

As a start, we apply our theories to some specific gates, including CNOT and SWAP, and derive the same results as from previous research~\cite{Piveteau2022circuit}. Since $\LOCC \subsetneq \SEP \subsetneq \PPT$, for single-shot exact $\gamma$-factor, we already have,
\begin{equation}\label{eq:gamma_factor_locc_sep_ppt}
    \gamma_{\LOCC}(\cN) \geq \gamma_{\SEP}(\cN) \geq \gamma_{\PPT}(\cN) 
\end{equation}
for any general bipartite channel $\cN$. Despite that, our result indicates that for a class of bipartite channels, the assistance of PPT-entanglement might contribute an insignificant advantage in reducing the sample complexity from the circuit knitting method. 

We particularly \textbf{remark} that for any two-qubit Clifford gates, there is no advantage in using the PPT-assisted QPD, and the $\gamma$-factors of the following bipartite gates are given by,
\begin{equation}
\begin{aligned}
    \gamma_{\cF}(\rm CNOT) &= 3;\\
    \gamma_{\cF}(i\rm{SWAP}) =
    \gamma_{\cF}(\rm SWAP) &= 7,
\end{aligned}
\end{equation}
where $\cF\in\{\LOCC, \SEP, \PPT\}$. Regarding the origin of circuit knitting, understanding the properties of two-qubit Clifford gates holds significant importance in terms of quantum advantages. Notably, it has been established that every such gate can be constructed via the gates $\mathbb{I}$, $\rm{CNOT}$, $\rm{SWAP}$, and  $i\rm{SWAP}$ up to local unitary gates. These lead to the unchanged PPT-assisted sampling overhead compared to the LOCC results from~\cite{Piveteau2022circuit}, and therefore, disproves the improvement from PPT-entanglement in circuit knitting on two-qubit Clifford gates.

To investigate the limitation of PPT-assisted knitting, we also involve the discussion on the parallel-cut scheme in this case, which has achieved benefits for cutting copies of a Clifford gate $U$ under LOCC~\cite{Piveteau2022circuit}. Consider cutting $n$-copies of CNOT gates, one can derive,
\begin{equation}
    \gamma_{\cF}(\Op{CNOT}^{\ox n}) = 2^{n+1} - 1
\end{equation}
for the same set of operations $\cF$ stated before. With the additional quantum memory for storing entanglement, one can demonstrate that the exact effective sampling cost of PPT-assisted knitting for the CNOT gate reaches the same asymptotic value, i.e., $\gamma^{\infty}_{\cF}(\Op{CNOT}) = \lim_{n\rightarrow \infty} \gamma^{(n)}_{\cF}(\Op{CNOT}) = 2$. The proof of the above remarks can refer to Appendix~\ref{appendix:proof_of_prop_gamma_ppt_lb_channel}. The corresponding CNOT evolution is called the bicovariant bidirectional channel~\cite{Bauml2018}. The max-Rains information bound in Eq.~\eqref{eq:max_rains_info_lb_regularized_gamma_ppt} reduces to the standard Rains relative entropy~\cite{Rains1999,Vedral1998} of its Choi state, indicating no benefits from the assistance of PPT operations under parallel cutting of CNOT.

\begin{figure*}[t]
    \centering
    \includegraphics[width=0.9\linewidth]{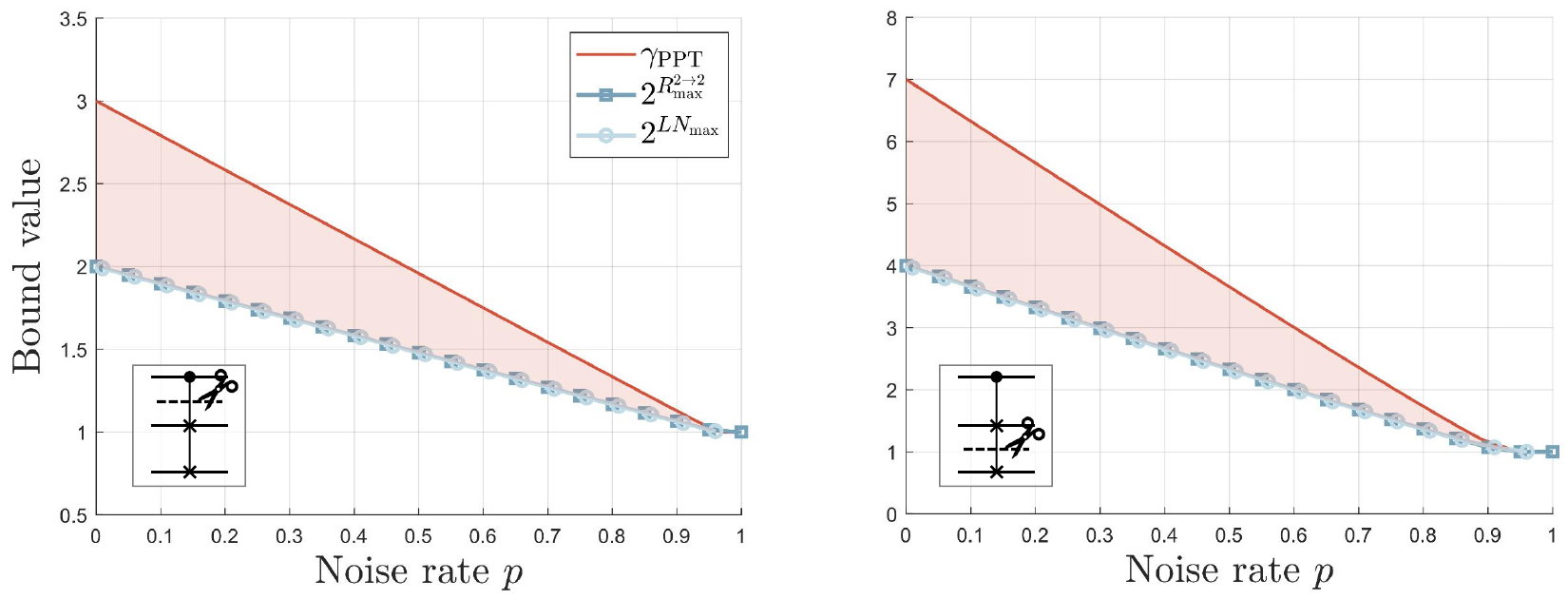}
    \caption{\textbf{Lower bounds comparison for noisy control-SWAP operation (CSWAP) undertaking the global depolarizing channel by the noise rate $p$ {with respect to different partitioning.}} {The Left and right figures illustrate the bound values of cutting qubit $1$ and $2$ and qubit $2$ and $3$, respectively.} The orange line gives the one-shot exact $\gamma_{\PPT}(\cN)$. The dark and light blue lines with square and circle markers illustrate the lower bounds of regularized PPT $\gamma$-factor by max-Rains information and the max logarithmic negativity, respectively. The shaded area indicates that our bounds serve as feasible efficient-computable lower bounds for $\gamma_{\PPT}^{\infty}(\cN)$.}
    \label{fig:cswap_cutingAB}
\end{figure*}

Furthermore, our bound works for general bipartite quantum channels with non-equal dimensionalities of $A$ and $B$, instead of only unitary channels studied with~\cite{Schmitt2023cutting} the \textit{KAK}-decomposition~\cite{Mitarai2021overhead}. It then serves as a useful tool for studying the nonlocality of non-KAK-like bipartite interactions and characterizing the optimal sampling overhead for non-unitary channels. For instance, the Toffoli gate is another important three-qubit gate discussed frequently in the quantum circuit architecture. Our observation showcases that, 
\begin{equation}
    \gamma^{(1)}_{\cF}(\rm Toffoli) = \gamma^{(2)}_{\cF}(\rm Toffoli) = 3,
\end{equation}
by setting a cut between the qubit 1-and-2 (1) or qubit 2-and-3 (2), as shown in~\cite{Schmitt2023cutting}, which acts as a tighter lower bound from the Choi state of the Toffoli gate. Interestingly, we have also tested the control-SWAP (CSWAP) gate, which is widely applied to quantum information tasks. In Fig.~\ref{fig:cswap_cutingAB}, compared to the results from Toffoli gate, the two cutting cases give very different values for the PPT-assisted $\gamma$-factor and the two lower bounds. Cutting between 1 and 2 qubits gives only half of the $\gamma$-factor from cutting 2 and 3 qubits. We observe that for general unitary cutting, the sampling cost may significantly vary by choosing different partitions.

\subsection{Exponential cost for knitting multiple instances of bipartite channels}\label{subsec:exponential_cost_for_knitting_multiple_instances}
For a general bipartite quantum channel $\cN_{AB\rightarrow A'B'}$, by the definition of parallel cutting, one can derive a lower bound for the $n$-effective sampling cost with finite $n$.
\begin{proposition}
    For any bipartite channel $\cN_{AB\rightarrow A'B'}$, the $n$-effective LOCC-assisted $\gamma$-factor can be lower bounded by
    \begin{equation}\label{eq:effective_gamma_lb}
        \gamma^{(n)}_{\LOCC}(\cN) \geq (2^{n E_{C,0}^{\PPT}(\cN)} - 1)^{1/n}.
    \end{equation}
\end{proposition}
Based on the proposition and the irreversible nature of the exact PPT-entanglement cost of bipartite quantum channels, we also have $\gamma^{(n)}_{\LOCC}(\cN)\geq (2^{n\cdot Q^{2\rightarrow 2}_{\LOCC}(\cN)} - 1)^{1/n}$ where $Q^{2\rightarrow 2}_{\LOCC}(\cN)$ is called the quantum communication capacity of $\cN$.

In the task of simulating $n\geq 1$ parallel instances of $\cN$, we obtain that the optimal number of samples of circuit knitting with respect to the LOCC channels scales at least $\Omega(4^{n E^{\PPT}_{C,0}(\cN)})$ delivering that the exponential growth in the sampling cost originates from the entanglement cost of its physical implementation. The number $n$ relates to the depth of the circuit, or the evolution time of the system, indicating the circuit knitting technique is generally inefficient for deep architectures. Besides, suppose the depth of the circuit is restricted. Only if $E^{\PPT}_{C,0}(\cN)$ scales $\cO(\log(N))$ where $N$ is the number of qubits in the system, the sampling cost can be polynomially lower bounded in $N$. In the large dimensional limit, only when $\cN$ requires negligible entanglement resource to be implemented the lower bound of sampling cost from circuit knitting can gain reasonable scaling for practical usage.

We then extend our discussion to the quantum circuit containing $n \geq 1$ distinct bipartite quantum channels or the noisy quantum gates, as a reasonable assumption on NISQ devices, denoted as $\cN_1, \cdots, \cN_n$. In the first scenario, each of these $\cN_j$'s is aligned to parallel via local swaps as local operations can preserve the $\gamma$-factor of each $\cN_j$. Given the aligned circuit as $\cN = \bigotimes_{j=1}^n \cN_j$, we directly derive the lower bound of $\gamma_{\cF}(\cN)$,
\begin{equation}\label{eq:parallel_cut_Nj}
    \gamma_{\cF}\left(\cN\right) \geq \left(\prod_{j=1}^n 2^{Q^{2\rightarrow 2}_{\LOCC}(\cN_j)}\right) - 1,
\end{equation}
where we observe that if every $\cN_j$ has insignificant entangling power, the product may tend to a constant scaling, and we may achieve an acceptable lower bound on $\gamma_{\cF}(\cN)$ for the practical implementation.

However, cutting the total circuit directly via the above is typically intractable as it requires tomography of $\cN$. Otherwise, it may require smart channel grouping strategies to keep each term in the product of Eq.~\eqref{eq:parallel_cut_Nj} sufficiently small. Instead, for each of these $n$ channels, suppose the corresponding $\cF$-assisted QPD has been pre-determined and denote the associated optimal sampling overheads by $\gamma_{\cF}(\cN_{j})$ for the $j$-th channel. Then, during each shot of execution on the circuit, we can only adopt the single cut strategy that $\cN_j$ gets independently and randomly replaced by one of its QPD channels, which leads to a total number of samples required by at least $\gamma^2_{\rm tot} = \prod_{j=1}^n \gamma_{\cF}(\cN_j)^2$. We then derive the following result.

\begin{corollary}\label{coro:gamma_tot}
    For a quantum circuit containing $n$ bipartite quantum channels $\cN_1, \cN_2, \cdots, \cN_n$ across two parties being cut,
    \begin{equation}
        \gamma_{\rm tot}^2 \geq \prod_{j=1}^n \left(2^{ E^{(1)}_{\cF,C,0}(\cN_j)} - 1\right)^2,
    \end{equation}
\end{corollary}
where $\cF \in \{\SEP, \PPT\}$.

Suppose not all channels have zero entanglement cost, and we denote the maximum one-shot exact $\cF$-entanglement cost among $\cN_j$'s as $\widehat{E}_{\cF,C}$. Then, in the large dimension limit, {the total sampling times via QPD can scale at least $\Omega( 4^{\widehat{E}_{\cF,C}})$} by the dominant entangling channel. 

If we assume the two parties of the system are isomorphic of dimension $d_A = d_B = d$, particularly for a qubit system, $d = 2^N$ where $N$ is the number of qubits in the half circuit. Notice that for those channels with a small entanglement cost, there may exist classical simulation tools, for example, tensor network method~\cite{Bridgeman2017}, which can efficiently simulate the dynamics of the system going beyond $50$ qubits, and hence lose the meaning of accessing real quantum devices. Only those with sufficient entangling power can truly showcase the potential quantum advantages and worth to be applied with circuit knitting. 

In fact, the entanglement cost of any bipartite channels can grow {linearly in $N$}~\cite{divincenzo1998quantum,Kianvash2022}. which then leads to the scaling of $\gamma_{\rm tot}^2$ as {$\Omega(4^{N})$}, which complements the results regarding the circuit depth~\cite{Piveteau2022circuit,Schmitt2023cutting}. Even with parallel cutting, we can derive that the total sampling cost based on the regularized form can scale at the same speed. Such an exponential growth in the sampling cost can not be avoided via the above strategies, which leads to the fundamental limitation of realizing the circuit knitting method on NISQ devices. From another perspective, the entangling channels are the quantum operations that can veritably showcase quantum advantages in computational tasks. Both the deterministic and the probabilistic simulation methods stated above via LOCC can be non-scalable and difficult.

\section{Concluding remarks}\label{sec:conclusion}
In this work, {we have addressed the principal challenge associated with the circuit knitting technique by rigorously demonstrating that the sampling cost is exponentially lower bounded by the entanglement measure—specifically, the exact entanglement cost, which elucidates the fact that more entanglement to implement the bipartite operation necessitates higher knitting costs.}

Specifically, our findings confirm an exponential relationship between the regularized $\gamma$-factor and the PPT- and SEP-entanglement costs of the nonlocal circuits. This relationship holds substantial implications, particularly as the dimensions increase, where the entangling capabilities of circuits across two parties become a significant factor. We have showcased through our analysis that the sampling cost scales as $\Omega(4^N)$ with respect to the number of qubits $N$ in the half circuit, a result relates to the circuit width that complements the findings from previous studies~\cite{Piveteau2022circuit} regarding the depth of the circuit. 

{Despite the potential for employing more powerful operations such as PPT and SEP channels, our research indicates that the use of the circuit knitting technique for simulating bidirectional interactions still results in an exponential growth in sampling costs, particularly when compared to scenarios assisted by LOCC. For instance, our proofs reveal that the $\gamma$-factors for LOCC-, PPT-, and SEP-QPD are identical for two-qubit Clifford gates, extending the existing literature that employs KAK decomposition for general bipartite channels~\cite{Schmitt2023cutting}.}

Furthermore, we have addressed concerns regarding the efficiency of existing methods by providing tighter lower bounds than those derived from the Schmidt decomposition, using the channel's Choi state as a reference point. Our numerical demonstrations provide clear, one-shot exact sample costs for the CSWAP gate across different partitionings, thereby illuminating the path forward in understanding the fundamental limitations imposed by the exponential sampling costs in circuit knitting relative to the entanglement cost of realizing any bipartite channel.

This research opens several avenues for future exploration. Notably, while we have begun to address the trade-offs between characterizing entire processes and the sampling costs associated with knitting these processes, more thorough investigations into the effects of composited channels are warranted. For example, the phenomenon where the composition of sequential CNOTs results in an identity gate suggests potential strategies for reducing sampling costs through smart gate grouping and circuit compiling techniques. On the other hand, the sampling cost with a certain simulation error tolerance may be more reasonable in a practical assumption of circuit knitting. The approximate version of our theory and also the connections to other resources are remained to solve. 

\section*{Acknowledgments}
The authors would like to thank Xuanqiang Zhao, Benchi Zhao, and  Zanqiu Shen for their valuable comments.
This work was partially supported by the National Key R\&D Program of China (Grant No.~2024YFE0102500), the Guangdong Provincial Quantum Science Strategic Initiative (Grant No.~GDZX2303007), the Guangdong Provincial Key Lab of Integrated Communication, Sensing and Computation for Ubiquitous Internet of Things (Grant No.~2023B1212010007),  the Start-up Fund (Grant No.~G0101000151) from HKUST (Guangzhou), the Quantum Science Center of Guangdong-Hong Kong-Macao Greater Bay Area, and the Education Bureau of Guangzhou Municipality.

\bibliography{main}

\appendix
\setcounter{subsection}{0}
\setcounter{table}{0}
\setcounter{figure}{0}

\vspace{2cm}
\onecolumngrid
\vspace{2cm}

\begin{center}
\textbf{
{\large{Appendix for `Circuit Knitting Faces Exponential Sampling Overhead\\ Scaling Bounded by Entanglement Cost'}}}
\end{center}

\numberwithin{equation}{section}
\renewcommand{\theproposition}{S\arabic{proposition}}
\renewcommand{\thedefinition}{S\arabic{definition}}
\renewcommand{\thefigure}{S\arabic{figure}}
\setcounter{equation}{0}
\setcounter{table}{0}
\setcounter{section}{0}
\setcounter{proposition}{0}
\setcounter{definition}{0}
\setcounter{figure}{0}


In these Supplementary Notes, we offer detailed proofs of the theorems and propositions in the manuscript. In Appendix~\ref{appendix:sdp_ppt_knitting_remarks}, we first provide details of semidefinite programming (SDP) for PPT-assisted circuit knitting in terms of the $\gamma$-factor. Based on the formulation, we then prove the remark for specific two-qubit gates. In Appendix~\ref{appendix:proof_of_regularized_ppt_sampling_cost} \& \ref{appendix:proof_of_regularized_sep_sampling_cost}, we prove our main theorems showing the exponential dependences between the regularized $\gamma$-factor and the exact entanglement cost of any given bipartite channels. In Appendix~\ref{appendix:proof_of_prop_gamma_ppt_lb_channel}, we also give another efficiently computable lower bound of the regularized PPT-assisted $\gamma$-factor using the bidirectional max-Rains information of the bipartite channels. In Appendix~\ref{appendix:lb_for_smoothed_regularized_locc_via_choi_state}, we introduce the concept of \textit{smoothed regularized LOCC-assisted $\gamma$-factor} of any bipartite channel and showcase that it is, in fact, lower bounded via the vanishing error entanglement cost of its Choi state in the asymptotic regime. In Appendix~\ref{appendix:circ_knit_distinct_channel}, we will give the analysis of the knitting cost for multiple distinct instances of bipartite quantum channels, including the proof of Corollary~\ref{coro:gamma_tot}.

\section{Notation and preliminaries}
Throughout the paper, we label different quantum systems by capital Latin letters, e.g., $A, B$. The respective Hilbert spaces for these quantum systems are denoted as $\cH_{A}$, $\cH_{B}$, each with dimension $d_A, d_B$. The set of all linear operators on $\cH_A$ is denoted by $\cL(\cH_A)$, with $I_A$ representing the identity operator. We denote by $\cL^{\dag}(\cH_A)$ the set of all Hermitian operators on $\cH_A$. In particular, we denote $\cD(\cH_A)\subseteq \cL^{\dag}(\cH_A)$ the set of all density operators being positive semidefinite and trace-one acting on $\cH_A$. A bipartite quantum state $\rho_{AB}\in \cD(\cH_A\ox\cH_B)$ is a linear operator acting on the product of Hilbert space $\cH_A\ox \cH_B$. The state $\rho_{AB}$ is called a positive partial transpose (PPT) state on the composite system $\cH_{AB}$ if $\cT_B(\rho_{AB}) = \rho_{AB}^{T_B}\geq 0$ where $\cT_B$ denotes the partial transpose operation regarding the system $B$. The set of separable (SEP) states on the composite system $\cH_{AB}$ (i.e. the states written at convex combinations of tensor product states) is a subset of all PPT states.

A quantum channel $\cN_{A\rightarrow A'}$ is a completely positive trace-preserving (CPTP) linear map that transforms linear operators from $\cL(\cH_{A})$ to $\cL(\cH_{A'})$. The Choi-Jamiołkowski operator of $\cN_{A\rightarrow A'}$ is expressed as $J^{\cN}_{AA'} = d_{A}(\cI_{\Bar{A}}\ox \cN)(\Phi_{\Bar{A}A}(d_{A}))$. We use  $\Tilde{J}^{\cN}_{AA'}$ to denote the corresponding Choi state of the channel. A bipartite channel $\cN_{AB\rightarrow A'B'}$ maps any linear operator in system $AB$ to system $A'B'$, i.e., $\cN_{AB\rightarrow A'B'}: \cL(\cH_{AB})\rightarrow \cL(\cH_{A'B'})$. In particular, a separable (SEP) bipartite channel is represented as~\cite{Wilde2013quantum},
\begin{equation}
    \cS_{AB\rightarrow A'B'} = \sum_j \cE_{A\rightarrow A'}^j \ox \cF_{B\rightarrow B'}^j,
\end{equation}
where $\{\cE_{A\rightarrow A'}^j\}$ and $\{\cF_{B\rightarrow B'}^j\}$ are sets of completely positive, trace-non-increasing maps (CPTN) such that $\cS_{AB\rightarrow A'B'}$ is trace preserving. Every SEP channel completely preserves the separability of quantum states and can not generate entanglement. An important example of a separable channel is the LOCC channel, which consists of local operations and classical communication. Further, a bipartite PPT channel completely preserves the PPT property of states, which satisfies $\cN^{\Gamma_B}\coloneqq \cT_{B'}\circ\cN\circ\cT_{B}$ is CPTP. This has been proven that the Choi matrix of $\cN$ must satisfy $(J^{\cN}_{ABA'B'})^{T_{BB'}}\succeq 0$ as for the state case.

The \textit{parallel cutting} setting in this work is first introduced in Ref.~\cite{Piveteau2022circuit} as shown in Fig.~\ref{fig:parallel_cut}. Imagine the situation of cutting the circuit containing $n$ instances of bipartite noisy bipartite gates $\cN_1, \cN_2, \cdots, \cN_n$. The optimal strategy of circuit knitting is to cut all the gates at the same time. To do so, one may have to first rearrange all the gates via local swap operations so that they become in parallel, i.e., $\cN = \bigotimes_{j=1}^n \cN_j$. The local swap operations can preserve the $\gamma$-factors as shown in~\cite{Piveteau2022circuit}.

\begin{figure*}[t]
    \centering
    \includegraphics[width=0.9\linewidth]{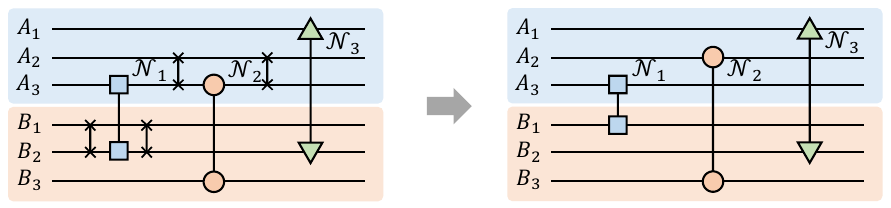}
    \caption{The strategy of aligning three bipartite noisy quantum gates to parallel via local SWAPs. $\cN_1, \cN_2$ and $\cN_3$ are originally acted on the qubit pairs $(A_3, B_2), (A_3, B_3)$ and $(A_1, B_2)$, respectively.}
    \label{fig:parallel_cut}
\end{figure*}

\section{SDP formulation and remarks from PPT-assisted circuit knitting}\label{appendix:sdp_ppt_knitting_remarks}
The definition of PPT-assisted circuit knitting resembles the original definition of LOCC-assisted circuit knitting. Given any bipartite quantum channel $\cN_{AB\rightarrow A'B'}$, one has to find a PPT-assisted quasi-probability decomposition (QPD) of it, i.e., $\cN = c_1 \cM_1 - c_2 \cM_2$ for some PPT channels $\cM_{1,2}$ and positive coefficients $c_{1,2}$. Then the $\gamma$-factor $\gamma_{\PPT}(\cN)$ can be estimated via the following semidefinite programming,
\begin{equation}\label{app_eq:SDP_PPT_channel_cost}
\begin{aligned}
    \gamma_{\PPT}(\cN) = \min_{c_{1,2}\geq 0,\ J_{\cM_{1,2}}} \;\;&  c_1 + c_2, \\
    {\rm s.t.} \;\;&  J^{\cM_{1,2}}_{ABA'B'}\succeq 0; \ (J^{\cM_{1,2}}_{ABA'B'})^{T_{BB'}}\succeq 0\\ 
    \;\;& J^{\cM_{1}}_{AB} = c_1 I_{AB};\ J^{\cM_{2}}_{AB} = c_2 I_{AB};\\
    \;\;& J^{\cN}_{ABA'B'} = J^{\cM_1}_{ABA'B'} - J^{\cM_2}_{ABA'B'}.
\end{aligned}
\end{equation}
For any arbitrary bipartite quantum channel $\cN$, i.e., the \textit{completely positive trace-preserving} (CPTP) maps, the last constraint in~\eqref{app_eq:SDP_PPT_channel_cost} guarantee $c_1 = c_2 + 1$ since for any state $\rho\in\cD(\cH_{AB})$, we have 
$$\tr[\cN(\rho)]  = c_1 \tr[\cM_1(\rho)] - c_2\tr[\cM_2(\rho)].$$
Therefore, we could rewrite the above SDP formulation for $\gamma_{\PPT}(\cN)$ by eliminating $c_2$ and also the $J^{\cM_2}_{ABA'B'}$. Writing $c = c_1$ and $J^{\cM}_{ABA'B'} = J^{\cM_1}_{ABA'B'}$, the objective function becomes $2c - 1$. The constraint $J^{\cM_2}_{ABA'B'} \succeq 0$ then becomes $J^{\cM}_{ABA'B'} \succeq J^{\cN}_{ABA'B'}$ and $\cT_{BB'}(J^{\cM_2}_{ABA'B'}) \succeq 0$ becomes $\cT_{BB'}(J^{\cM}_{ABA'B'})\succeq \cT_{BB'}(J^{\cN}_{ABA'B'})$. Notice that the second \textit{trace-preserving} constraint then becomes automatically satisfied, and we derive the equivalent SDP for $\gamma_{\PPT}(\cN)$,
\begin{equation}\label{app_eq:modified_SDP_PPT_channel_cost}
\begin{aligned}
    \gamma_{\PPT}(\cN) = \min_{c\geq 0,\ J_{\cM}} \;\;&  2c-1, \\
    {\rm s.t.} \;\;& (J^{\cM}_{ABA'B'})^{T_{BB'}}\succeq 0\\ 
    \;\;& J^{\cM}_{AB} = c I_{AB};\\
    \;\;& J^{\cM}_{ABA'B'}\succeq J^{\cN}_{ABA'B'};\\
    \;\;& (J^{\cM}_{ABA'B'})^{T_{BB'}}\succeq (J^{\cN}_{ABA'B'})^{T_{BB'}},
\end{aligned}
\end{equation}
Notice that the optimal value of $c$ from the above optimization is called the \textit{robustness of dynamical resource}~\cite{Gour2019how,regula2021fundamental}, we have $\gamma_{\PPT}(\cN) = 2R_{\PPT}(\cN) - 1$. The value of $R_{\PPT}(\cN)$ is closely related to the generalized robustness of PPT-entanglement~\cite{vidal1999robustness,regula2021fundamental}. Notice that PPT operations are stronger than LOCC operations, this leads to $\gamma_{\LOCC}(\cN) \geq \gamma_{\PPT}(\cN)$. Therefore, for those two-qubit gates in the \textbf{Remark}, one direction has been proven referring to~\cite{Piveteau2022circuit}, which gives an equivalence argument for LO- and LOCC-assisted $\gamma$-factor. It suffices to prove the following,
\begin{equation}\label{app_eq:lowerbound_gamma_ppt_special_gates}
\begin{aligned}
    \gamma_{\PPT}(\rm CNOT) &\geq 3; \\ 
    \gamma_{\PPT}(i\rm{SWAP}) &\geq 7; \\ 
    \gamma_{\PPT}(\rm SWAP) &\geq 7.
\end{aligned}
\end{equation}
We prove these by first defining a new quantity $\widehat{W}(\cN)$ which is derived from taking a relaxation on the programming of $R_{\PPT}(\cN)$ by removing the constraint $\cT_{BB'}(J^{\cM}_{ABA'B'}) \succeq \cT_{BB'}(J^{\cN}_{ABA'B'})$. The optimal value, which we denoted as $\widehat{W}(\cN)$, can be estimated via the following SDP programming,
\begin{equation}\label{app_eq:primal_W_hat_channel}
\begin{aligned}
    \widehat{W}(\cN) = \min_{c\geq 0,\ J_{\cM}} \;\;&  c, \\
    {\rm s.t.} \;\;&  (J^{\cM}_{ABA'B'})^{T_{BB'}}\succeq 0;\\
    \;\;& J^{\cM}_{ABA'B'}\succeq J^{\cN}_{ABA'B'};\\ 
    \;\;& J^{\cM}_{AB} \preceq c I_{AB}.
\end{aligned}
\end{equation}
Clearly, we have, $R_{\PPT}(\cN) \geq \widehat{W}(\cN)$, therefore, $\gamma_{\PPT}(\cN) \geq 
2\widehat{W}(\cN) - 1$. We further derive the dual programming of~\eqref{app_eq:primal_W_hat_channel} follows by writing out the \textit{Lagrangian} of the primal programming of $\widehat{W}(\cN)$ as,
\begin{equation}
\begin{aligned}
    L(P,Q,R) &= c - \tr[P(J^{\cM}_{ABA'B'})^{T_{BB'}}] - \tr[Q(J^{\cM}_{ABA'B'} - J^{\cN}_{ABA'B'})]\\ &\quad+ \tr[R(J^{\cM}_{AB} - cI_{AB})],\\
    &= \tr[QJ^{\cN}_{ABA'B'}] + c(1 - \tr[R]) + \tr[J^{\cM}_{ABA'B'}(R\ox I_{A'B'} - P^{T_{BB'}} - Q)] 
\end{aligned}
\end{equation}
where $P,Q,R\succeq 0$ are the introduced Lagrangian multipliers and $P, Q\in\cL^{\dag}(\cH_{ABA'B'})$, $R\in\cL^{\dag}(\cH_{AB})$. To ensure the dual function non-trivially lower bounded, we have $Q\preceq R\ox I_{A'B'} - P^{T_{BB'}}$ and $\tr[R]\leq 1$. Taking the change of variables $Y = R\ox I_{A'B'} - P^{T_{BB'}}$ and due to Slater's condition, the strong duality holds, and we have the following,
\begin{equation}
\begin{aligned}
    \widehat{W}(\cN) = \max_{Y, Q\succeq 0, R} \;\; &\tr(QJ^{\cN}_{ABA'B'})\\
    {\rm s.t.}   \;\; 
    & Q \preceq Y; \ Q, R\succeq 0; \tr[R] \leq 1;\\
    \;\; 
    & Y^{T_{BB'}} \preceq R^{T_B} \ox I_{A'B'}.
\end{aligned}
\end{equation}
Notice that $Q$ can be now seen as a slackened variable by requiring $Y = Q$, and we, therefore, have the equivalent SDP,
\begin{equation}\label{app_eq:dual_W_hat_channel}
\begin{aligned}
    \widehat{W}(\cN) = \max_{Y, R\succeq 0} \;\; &\tr(Y_{ABA'B'}J^{\cN}_{ABA'B'})\\
    {\rm s.t.}  \;\; 
    & \tr[R_{AB}] \leq 1;\\
    \;\; 
    & Y_{ABA'B'}^{T_{BB'}} \preceq R_{AB}^{T_B} \ox I_{A'B'}.
\end{aligned}
\end{equation}
Now we are ready to prove the inequalities in~\eqref{app_eq:lowerbound_gamma_ppt_special_gates}. Start with the CNOT operation. Notice that we could take $Y = J^{\cN}_{ABA'B'}/8$ and $R_{AB} = I_{AB}/d_{AB}$ as a feasible solution to~\eqref{app_eq:dual_W_hat_channel} since $\tr[R] = 1$ and $|\operatorname{eig}[Y_{ABA'B'}^{T_{BB'}}]| = |\operatorname{eig}[(J^{\cN}_{ABA'B'})^{T_{BB'}}]| / 8 \leq 1/4$. Therefore, the last constraint from~\eqref{app_eq:dual_W_hat_channel} holds and we have,
\begin{equation}
    \widehat{W}(\cN) \geq \frac{1}{8}\tr[(J^{\cN}_{ABA'B'})^2] = \frac{16}{8} = 2,
\end{equation}
where $\tr[(J^{\cN}_{ABA'B'})^2] = 4\tr[J^{\cN}_{ABA'B'}] = 16$ as $J^{\cN}_{ABA'B'}$ is unnormalized pure state for CNOT. We then have proven the statement for CNOT as $\gamma_{\PPT}(\rm{CNOT}) \geq 4 - 1 = 3$. Similarly, we could take $Y = J^{\cN}_{ABA'B'} / 4, R_{AB} = I_{AB}$ as feasible solutions to $\rm{SWAP}$ and $i\rm{SWAP}$ as $|\operatorname{eig}[Y_{ABA'B'}^{T_{BB'}}]| = |\operatorname{eig}[(J^{\cN}_{ABA'B'})^{T_{BB'}}]| / 4 \leq 1$, and therefore finish proofs for the remark.

\section{Proof of Theorem for regularized PPT-assisted sampling cost}~\label{appendix:proof_of_regularized_ppt_sampling_cost}
In order to prove the exponential lower bound for the regularized PPT-assisted $\gamma$-factor with respect to the exact PPT-entanglement cost, we first prove the following lemma for the one-shot scenario.
\begin{lemma}\label{app_lem:gammaPPT_1shot}
For a bipartite quantum channel $\cN_{AB\rightarrow A'B'}$,
\begin{equation}
    \gamma_{\PPT}(\cN_{AB\rightarrow A'B'}) \geq 2^{LN_{\max }\left(\cN_{AB\rightarrow A'B'}\right)} - 1.
\end{equation}
\end{lemma}
\begin{proof}
Recalling the definition of $\gamma_{\PPT}(\cN_{AB\rightarrow A'B'})$ which can be re-expressed as
\begin{equation}
\begin{aligned}
    \gamma_{\PPT}(\cN_{AB\rightarrow A'B'})= \min & \; 2c - 1\\
    {\rm s.t.} \, &\; J_{ABA'B'}^{\cN} = cJ_{ABA'B'}^{\cM_1} - (c-1)J_{ABA'B'}^{\cM_2},\\
                  &\; J_{ABA'B'}^{\cM_1} \geq 0, \, J_{ABA'B'}^{\cM_2} \geq 0,\\
                  &\; (J_{ABA'B'}^{\cM_1})^{T_{BB'}} \geq 0, \, (J_{ABA'B'}^{\cM_2})^{T_{BB'}} \geq 0,\\
                  &\; J_{AB}^{\cM_1} = I_{AB},\, J_{AB}^{\cM_2} = I_{AB}.
\end{aligned}
\end{equation}
Suppose $\{J_{ABA'B'}^{\cM_1}, J_{ABA'B'}^{\cM_2}, c\}$ is a feasible solution for $\gamma_{\PPT}(\cN)$. Then we have 
\begin{equation}
    (J_{ABA'B'}^{\cN})^{T_{BB'}} + (c-1)(J_{ABA'B'}^{\cM_2})^{T_{BB'}} \geq 0, \; c(J_{ABA'B'}^{\cM_1})^{T_{BB'}} - (J_{ABA'B'}^{\cN})^{T_{BB'}} \geq 0
\end{equation}
which yields
\begin{equation}\label{Eq:JN_ineq1}
    -(c-1)(J_{ABA'B'}^{\cM_2})^{T_{BB'}} \leq (J_{ABA'B'}^{\cN})^{T_{BB'}} \leq c(J_{ABA'B'}^{\cM_1})^{T_{BB'}}.
\end{equation}
By the fact that $(J_{ABA'B'}^{\cM_1})^{T_{BB'}} \geq 0$ and $(J_{ABA'B'}^{\cM_2})^{T_{BB'}}\geq 0$, we have
\begin{equation}
    -(c-1)(J_{ABA'B'}^{\cM_1}+J_{ABA'B'}^{\cM_2})^{T_{BB'}} \leq (J_{ABA'B'}^{\cN})^{T_{BB'}} \leq c(J_{ABA'B'}^{\cM_1}+J_{ABA'B'}^{\cM_2})^{T_{BB'}}
\end{equation}
which further gives
\begin{equation}
    -c(J_{ABA'B'}^{\cM_1}+J_{ABA'B'}^{\cM_2})^{T_{BB'}} \leq (J_{ABA'B'}^{\cN})^{T_{BB'}} \leq c(J_{ABA'B'}^{\cM_1}+J_{ABA'B'}^{\cM_2})^{T_{BB'}}.
\end{equation}
Let $P_{ABA'B'} = c(J_{ABA'B'}^{\cM_1} + J_{ABA'B'}^{\cM_2})$. We have that $P_{ABA'B'}$ is a feasible solution for $LN_{\max }\left(\cN_{AB\rightarrow A'B'}\right)$ and notice that by taking the partial trace on the subsystem $A'B'$, we have,
\begin{equation}
    P_{AB} = P_{AB}^{T_B}= c(J_{AB}^{\cM_1} + J_{AB}^{\cM_1}) = 2c \cdot I_{AB}.
\end{equation}
Now we know that for each feasible solution $\{J_{ABA'B'}^{\cM_1}, J_{ABA'B'}^{\cM_2}, c\}$ of $\gamma_{\PPT}(\cN_{AB\rightarrow A'B'})$ with an objective value $2c-1$, there exists a feasible solution for $LN_{\max}(\cN_{AB\rightarrow A'B'})$ that gives an objective value less or equal to $\log(2c)$. Therefore, we have,
\begin{equation}
\begin{aligned}
    &2^{LN_{\max}(\cN_{AB\rightarrow A'B'})}\leq 2c\\
    \Rightarrow & \ \gamma_{\PPT}(\cN_{AB\rightarrow A'B'}) \geq 2^{LN_{\max}(\cN_{AB\rightarrow A'B'})} - 1.
\end{aligned}
\end{equation}
By the fact that $E^{\PPT}_{C,0}(\cN_{AB\rightarrow A'B'}) \leq LN_{\max}(\cN_{AB\rightarrow A'B'})$, we can also conclude that 
\begin{equation}
    \gamma_{\PPT}(\cN_{AB\rightarrow A'B'}) \geq 2^{E^{\PPT}_{C,0}(\cN_{AB\rightarrow A'B'})} - 1.
\end{equation}
\end{proof}
\begin{remark}
    Since the exact entanglement cost of NPT-entanglement is no smaller than the exact distillable NPT-entanglement, which is no smaller than NPT entanglement generating power~\cite{Gour2020,Bauml2018,Bennett2003}, i.e., 
    \begin{equation}
        E^{\PPT}_{D}(\cN_{AB\rightarrow A'B'}) \leq E_{C}^{\PPT}(\cN_{AB\rightarrow A'B'})
    \end{equation}
    we can obtain similar lower bounds for $\gamma_{\PPT}(\cN_{AB\rightarrow A'B'})$ correspondingly up to a constant factor.
\end{remark}

\begin{theorem}
    For a bipartite quantum channel $\cN$, the LOCC-assisted regularized $\gamma$-factor is lower bounded by the exponential of the (parallel) NPT-entanglement cost of $\cN$,
    \begin{equation}
        \gamma^{\infty}_{\PPT}(\cN_{AB\rightarrow A'B'}) \geq 2^{E^{\PPT}_{C,0}(\cN_{AB\rightarrow A'B'})}.
    \end{equation} 
\end{theorem}
\begin{proof}
    The proof of the theorem is stated as follows. For $n$-shot scenario, the above proof of Lemma~\ref{app_lem:gammaPPT_1shot} provides a lower bound of simulating $n$ parallel $\cN_{AB\rightarrow A'B'}$ as,
    \begin{equation}
        \gamma_{\PPT}(\cN^{\ox n}) \geq 2^{LN_{\max}(\cN^{\ox n})} - 1.
    \end{equation}
    By~\cite{Bauml2019resource,Gour2021entanglement}, we also have $LN_{\max}$ is super-additive, which therefore implies 
    \begin{equation}\label{eq:LN_max_nshot}
        \gamma_{\PPT}(\cN^{\ox n}) \geq 2^{nLN_{\max}(\cN)} - 1.  
    \end{equation}
    Then, by taking $n$-th root of both sides of the inequality and applying the limitation as $n\rightarrow \infty$, we could derive,
    \begin{equation}
    \begin{aligned}
        \log(\gamma^{\infty}_{\PPT}(\cN)) &= \lim_{n\rightarrow \infty}\log(\gamma^{(n)}_{\PPT}(\cN)) \geq \lim_{n\rightarrow \infty} \frac{1}{n}\log\left(2^{nLN_{\max}(\cN)} - 1\right) = LN_{\max}(\cN)\\
        &\Rightarrow \gamma^{\infty}_{\PPT}(\cN) \geq 2^{LN_{\max}(\cN)} \geq 2^{E^{\PPT}_{C,0}(\cN)},
    \end{aligned}
    \end{equation}
    as desired.
\end{proof}

\section{Proof of Theorem for regularized SEP-assisted sampling cost}\label{appendix:proof_of_regularized_sep_sampling_cost}
In order to prove the exponential lower bound for the regularized SEP-assisted $\gamma$-factor with respect to the exact SEP-entanglement cost, we first prove the following lemma for the one-shot scenario.
\begin{lemma}\label{app_lem:gammaSEP_1shot}
For a bipartite quantum channel $\cN_{AB\rightarrow A'B'}$,
\begin{equation}
    \gamma_{\SEP}(\cN_{AB\rightarrow A'B'}) \geq 2^{E_{\SEP,C,0}^{(1)}(\cN_{AB\rightarrow A'B'})} - 1.
\end{equation} 
\end{lemma}
\begin{proof}
Suppose $\gamma_{\SEP}(\cN) = 2c+1$ with the optimal decomposition, $\cN = (c+1)\cN_1 - c\cN_2$ where $\cN_1, \cN_2\in \mathrm{SEPP}(A:B)$. Now we construct a map $\Lambda$ on system $A\Bar{A}B\Bar{B}$ as follows, where $|\Bar{A}| = |\Bar{B}| = k$ and $k:=2\lceil c \rceil$ s.t., any input state $\sigma\in\cD(\cH_{A\Bar{A}}\ox\cH_{B\Bar{B}})$,
\begin{equation}
    \Lambda(\sigma_{A\Bar{A}B\Bar{B}}) = \tr[\Phi_k\sigma_{\Bar{A}\Bar{B}}]\cN(\sigma_{AB}) + \tr[(I-\Phi_{k})\sigma_{\Bar{A}\Bar{B}}] \cN_2(\sigma_{AB}),
\end{equation}
where we denote $\sigma_{AB} := \tr_{\Bar{A}\Bar{B}}\sigma_{A\Bar{A}B\Bar{B}}$ and $\sigma_{\Bar{A}\Bar{B}} := \tr_{AB}\sigma_{A\Bar{A}B\Bar{B}}$, and denote by $\Phi_k$ the maximally entangled state in system $\Bar{A}\Bar{B}$. In the following, we shall show that $\Lambda\in \mathrm{SEPP}(A\Bar{A}:B\Bar{B})$. Note that for all $\sigma_{A\Bar{A}B\Bar{B}}\in\SEP(A\Bar{A}:B\Bar{B})$, we can write $\sigma_{A\Bar{A}B\Bar{B}} = \sum_{j}p_j\sigma_{A\Bar{A}}^{(j)}\ox\sigma_{B\Bar{B}}^{(j)}$. It follows that $\sigma_{\Bar{A}\Bar{B}} = \sum_{j}p_j\sigma_{\Bar{A}}^{(j)}\ox\sigma_{\Bar{B}}^{(j)}$ is a separable state in system $\Bar{A}\Bar{B}$, so is $\sigma_{AB}$. Then, by the property of the maximally entangled state, we have $\tr[\Phi_k\sigma_{\Bar{A}\Bar{B}}]\leq 1/k$. Therefore, if we denote the measurement probability $p = \tr[\Phi_k \sigma_{\Bar{A}\Bar{B}}]$ and $1-p = \tr[(I_{AB} - \Phi_k) \sigma_{\Bar{A}\Bar{B}}]$ for all $\sigma_{A\Bar{A}B\Bar{B}} \in \SEP(A\Bar{A}:B\Bar{B})$, it holds that
\begin{equation}
    \Lambda(\sigma_{A\Bar{A}B\Bar{B}}) = p\cN(\sigma_{AB}) + (1-p) \cN_2(\sigma_{AB}),
\end{equation}
for some $p\leq 1/k$, which leads to $(1-p)/p \geq k-1 = 2\lceil c \rceil - 1 \geq c$ as $c\geq 1$, and we have $p\cN + (1-p) \cN_2 \in \mathrm{SEPP}(A:B)$ by the definition of $\gamma_{\SEP}(\cN)$. Thus, $\Lambda(\sigma_{A\Bar{A}B\Bar{B}})\in \SEP(A\Bar{A}:B\Bar{B})$ which indicate that $\Lambda\in \mathrm{SEPP}(A\Bar{A}:B\Bar{B})$. 
Now consider input states having form $\rho_{AB}\ox \Phi_k$. We have for $\Lambda\in \mathrm{SEPP}(A\Bar{A}:B\Bar{B})$,
\begin{equation}
\Lambda(\rho_{AB}\ox \Phi_k) = \cN(\rho_{AB}), \; \forall \rho_{AB}\in \cD(\cH_A\ox\cH_B),
\end{equation}
which then shows that $\Lambda$ forms a one-shot SEP simulation protocol of $\cN$ consuming $\Phi_k$. By the definition of the exact SEP-entanglement cost of $\cN$, we have $E_{\SEP,C,0}^{(1)}(\cN)\leq \log k$. Therefore, by $2c+1\geq 2\lceil c\rceil-1$, we have
\begin{equation}
    \gamma_{\SEP}(\cN) \geq 2^{E_{\SEP,C,0}^{(1)}(\cN)} - 1.
\end{equation}
We also want to remark the above proof follows a similar idea from Ref.~\cite{Chitambar2017entanglement}.
\end{proof}

Besides, our upper bound can be directly derived using the results from Ref.~\cite{Kim2021one} about one-shot dynamic entanglement manipulation.
\begin{equation}
    E^{(1)}_{\SEP, C, 0}(\cN) \leq \log(R_s^{\SEP}(\cN)) + 2.
\end{equation}
By simply exponentiating the inequality and re-arranging the terms, we can also derive,
\begin{equation}
    2^{E^{(1)}_{\SEP, C, 0}(\cN)} \leq 4R_s^{\SEP}(\cN) = 2(\gamma_{\SEP}(\cN)-1) \Rightarrow \gamma_{\SEP}(\cN) \geq 2^{E^{(1)}_{\SEP, C, 0}(\cN) - 1} + 1.
\end{equation}
Both of the above attempts can lead to the lower bound for the regularized sampling overhead of any bipartite channels.

\begin{corollary}
    For any bipartite quantum channel $\cN_{AB\rightarrow A'B'}$,
    \begin{equation}
        \gamma^{\infty}_{\SEP}(\cN_{AB\rightarrow A'B'}) \geq 2^{E^{\SEP}_{C,0}(\cN_{AB\rightarrow A'B'})}.
    \end{equation}
\end{corollary}    
\begin{proof}
Firstly, from the faithfulness, we know that $E^{\SEP}_{C,0}(\cN)=0$ and $\gamma_{\SEP}^{\infty}(\cN) \geq 1=2^{E^{\SEP}_{C,0}(\cN)}$ for all $\cN\in \mathrm{SEPP}(A:B)$. Next, as long as the minimum dimension of the consumed entanglement resource must be an integer, we have $E_{\SEP,C,0}^{(1)}(\cN)\geq 1$, for all $\cN\notin \mathrm{SEPP}(A:B)$ and, therefore, $2^{E_{\SEP,C,0}^{(1)}(\cN)} - 1 \geq 2^{E_{\SEP,C,0}^{(1)}(\cN)-1}$. Then fix $n\geq 1$. By Lemma~\ref{app_lem:gammaSEP_1shot}, we have
\begin{equation}\label{Eq:gamma_SEP_infy}
    \gamma^{(n)}_{\SEP}(\cN)= (\gamma_{\SEP}(\cN^{\ox n}))^{\frac{1}{n}} \geq (2^{E_{\SEP,C,0}^{(1)}(\cN^{\ox n})-1})^{\frac{1}{n}}= 2^{\frac{1}{n}E_{\SEP,C,0}^{(1)}(\cN^{\ox n})-\frac{1}{n}}.
\end{equation}
Taking the limit $n\rightarrow \infty$ on both sides of Eq.~\eqref{Eq:gamma_SEP_infy} and applying the continuity of exponential function, we have
\begin{equation}
    \gamma_{\SEP}^{\infty}(\cN) = \lim_{n\rightarrow \infty}\gamma^{(n)}_{\SEP}(\cN) \geq 2^{\lim\limits_{n\rightarrow \infty}\left(\frac{1}{n}E_{\SEP,C,0}^{(1)}(\cN^{\ox n})-\frac{1}{n}\right)}= 2^{E_{C,0}^{\SEP}(\cN)}.
\end{equation}
\end{proof}

\section{Proof of the bidirectional max-Rains information lower bound of the regularized sampling cost}~\label{appendix:proof_of_prop_gamma_ppt_lb_channel}
Stands on another point of view for finding an additive lower bound for one-shot exact $\gamma_{\PPT}(\cN)$, one can recall the bidirectional max-Rains information of any bipartite channel and derive the following lemma,
\begin{lemma}\label{app_lem:k_copy_ppt_lower_bound_for_channel}
    The PPT-assisted $\gamma$-factor for $k$-parallel cutting of a bipartite channel $\cN_{AB \rightarrow A'B'}$ is lower bounded by,
    \begin{equation}
        \gamma_{\PPT}(\cN^{\otimes n}) \geq 2^{n\cdot R^{2\rightarrow 2}_{\max}(\cN) + 1} - 1,
    \end{equation}
    where $R^{2\rightarrow 2}_{\max}(\cN)$ is bidirectional max Rains information of $\cN$.
\end{lemma}
\begin{proof}
Let us start with a one-shot situation. Given a quantum channel $\cN$, recalling the above relations as $\gamma_{\PPT}(\cN) \geq 2\widehat{W}(\cN) - 1$ where $\widehat{W}(\cN)$ can be evaluated via SDP~\eqref{app_eq:dual_W_hat_channel}. Observing that the definition of bidirectional max-Rains information from~\cite{Bauml2018, Wang2018semidefinite} of any  bipartite channel $\cN$, i.e., $R^{2\rightarrow 2}_{\max}(\cN)\coloneqq \log\Gamma^{2\rightarrow 2}(\cN)$ where $\Gamma^{2\rightarrow 2}(\cN)$  is evaluated by the following dual SDP programming,
\begin{equation}
\begin{aligned}
    \Gamma^{2\rightarrow 2}(\cN) = \max_{Y, R\succeq 0}\Big\{\tr[J^{\cN}_{ABA'B'} Y], \tr[R_{AB}] = 1, -R_{AB}\ox I_{A'B'} \preceq Y_{ABA'B'}^{T_{BB'}} \preceq R_{AB}\ox I_{A'B'} \Big\},
\end{aligned}
\end{equation}
which can be seen as a relaxation of the programming~\eqref{app_eq:dual_W_hat_channel} for $\widehat{W}(\cN)$. Combine the above, we have,
\begin{equation}
    R_{\PPT}(\cN) \geq \widehat{W}(\cN) \geq \Gamma^{2\rightarrow 2}(\cN),
\end{equation}
and as a result, for the single-copy case, 
\begin{equation}\label{app_eq:gamma_ppt_lb_single_cut}
    \gamma_{\PPT}(\cN) \geq 2\Gamma^{2\rightarrow 2}(\cN) - 1 = 2\cdot 2^{R^{2\rightarrow 2}_{\max}(\cN)} - 1.
\end{equation}
\end{proof}

Now suppose the entanglement factory is open, and we are allowed to use PPT-entanglement while simulation for the $n$-call situation, 
based on the additivity~\cite{Wang2018semidefinite} of the max-Rains information, we could prove that for the situation of cutting $n$ parallel uses of $\cN$.
\begin{equation}
    \gamma_{\PPT}(\cN^{\otimes n}) \geq 2\cdot 2^{R^{2\rightarrow 2}_{\max}(\cN^{\otimes n})} - 1 = 2^{nR^{2\rightarrow 2}_{\max}(\cN) + 1} - 1.
\end{equation}
Another direct observation from the above relation, taking $\rm CNOT$ gate as an example, is that, 
\begin{equation}
    \gamma_{\LOCC}(\operatorname{CNOT}^{\ox n}) = \gamma_{\PPT}(\operatorname{CNOT}^{\ox n}) = 2^{n+1} - 1,
\end{equation}
which has combined with the LOCC results in~\cite{Piveteau2022circuit}. The `$\geq$' direction can be proven using the feasible solution provided previously for the dual problem of $\widehat{W}(\cN)$, which also satisfies the constraints for $\Gamma^{2\rightarrow 2}(\cN)$. Hence, for CNOT, $R^{2\rightarrow 2}_{\max}(\rm{CNOT}) \geq 1$ and we have,
\begin{equation}
    \gamma_{\PPT}(\operatorname{CNOT}^{\ox n}) \geq 2^{n+1} - 1.
\end{equation}

The proof for the regularized result is based on Lemma~\ref{app_lem:k_copy_ppt_lower_bound_for_channel}. Suppose now an infinite amount of PPT entanglement was supplied by the factory. We can then derive the following theorem,
\begin{theorem}\label{app_them:regularized_max_rains_bound}
    For any fixed bipartite quantum channel $\cN_{AB\rightarrow A'B'}$, the PPT-assisted regularized $\gamma$-factor is exponentially lower bounded by its bidirectional max-Rains information, i.e.,
    \begin{equation}
        \gamma^{\infty}_{\PPT}(\cN) \geq 2^{R_{\max}^{2\rightarrow 2}(\cN)}.
    \end{equation}
\end{theorem}
\begin{proof}
We prove the theorem starting with the lower bound of the $n$-copy PPT-assisted $\gamma$-factor. Given the situation of $n$-parallel cutting of a bipartite channel $\cN_{AB\rightarrow A'B'}$, the monotonicity of $ \sqrt[n]{\ \cdot \ }$ leads to,
\begin{equation}
    (\gamma_{\PPT}(\cN^{\otimes n}))^{1/n} \geq (2\Gamma^{2\rightarrow 2}(\cN)^n - 1)^{1/n},
\end{equation}
for any positive integer $n$. Therefore, taking the limit with respect to $n$ would preserve the inequality as,
\begin{equation}
\begin{aligned}
    \gamma^{\infty}_{\PPT}(\cN) \geq \lim_{n\rightarrow \infty}\left(2\Gamma^{2\rightarrow 2}(\cN)^n - 1\right)^{1/n}.
\end{aligned}
\end{equation}
Denote $L = \lim_{n\rightarrow \infty}\left(2\Gamma^{2\rightarrow 2}(\cN)^n - 1\right)^{1/n}$ and by the continuity of logarithm, we could interchange the order of limit and logarithm and hence derive the following by L'Hopital's rule,
\begin{equation}
\begin{aligned}
    \log L &= \lim_{k\rightarrow \infty}\frac{1}{k}\log\left(2\Gamma^{2\rightarrow 2}(\cN)^k - 1\right) =  R^{2\rightarrow 2}_{\max}(\cN)\\ &\Rightarrow \gamma^{\infty}_{\PPT}(\cN) \geq L = 2^{R^{2\rightarrow 2}_{\max}(\cN)}.
\end{aligned}
\end{equation}
\end{proof}

Notice that, since $\gamma^{\infty}_{\LOCC}(\cN) \geq \gamma^{\infty}_{\PPT}(\cN)$ as $\PPT$ is a relaxation of $\LOCC$ constraints. Besides, from previous literature~\cite{Bauml2018}, we have $R^{2\rightarrow 2}_{\max}(\cN) \geq Q_{\PPT}^{2\rightarrow 2}(\cN) \geq Q^{2\rightarrow 2}_{\LOCC}(\cN)$. Therefore, we have demonstrated the chain of inequalities as
\begin{equation}
    \gamma^{\infty}_{\LOCC}(\cN) \geq \gamma^{\infty}_{\PPT}(\cN) \geq 2^{Q^{2\rightarrow 2}_{\LOCC}(\cN)}.
\end{equation}

\section{Lower bound for smoothed regularized LOCC sampling cost via entanglement cost of Choi state}\label{appendix:lb_for_smoothed_regularized_locc_via_choi_state}
Apart from the exact scenario of circuit knitting, with the spirit of quantum reversed Shannon theory, by treating the $\gamma$-factor as a resource measure, one can define the vanishing-error version of the regularized $\gamma$-factor in the asymptotic regime. For any target bipartite channel $\cN_{AB\rightarrow A'B'}$, given $\cF$ a convex free operation set regarding a specific resource theory, the decomposition protocol reduces two terms as $\{(c_1, \cM_1),(c_2, \cM_2)\}$. In the practical situation, the realization of any quantum operation may yield unavoidable systematic error~\cite{Endo2018practical}. By allowing an $\epsilon$ error tolerance, we define the \textit{smoothed}  $\gamma$-factor of $\cN$ via the assistance of operations in $\cF$ as,
\begin{equation}
\begin{aligned}
    \gamma_{\cF}^{\epsilon}(\cN)&\coloneqq \min \{c_1 + c_2 \; \cM = c_1\cM_1 - c_2 \cM_2, \ c_{1,2} \geq 0, \ \cM_{1,2}\in \cF, \cM \in B_{\epsilon}(\cN) \},
\end{aligned}
\end{equation}
where the $\epsilon$-Ball of $\cN$ regarding the diamond norm, i.e.,
\begin{equation}
    B_{\epsilon}(\cN) = \{\cM \in \operatorname{CPTP}(AB\rightarrow A'B') \mid \|\cN - \cM\|_{\diamond} \leq \epsilon\}.
\end{equation}
Once taking $\epsilon \rightarrow 0$, the above definition becomes the usual one-shot definition of exact $\gamma$-factor regarding the free operation set $\cF$, i.e., $\gamma_{\cF}(\cN) = \lim_{\epsilon \rightarrow 0} \gamma_{\cF}^{\epsilon}(\cN)$. The approximate sampling cost is also generally discussed in~\cite{Piveteau2022quasiprobability}. We can further define the smoothed regularized $\gamma$-factor with respect to $\cF$ in the following,
\begin{definition}
    Given any (bipartite) quantum channel $\cN$ and a free operation set $\cF$, the smoothed regularised $\gamma$-factor of $\cN$ with respect to $\cF$ is defined as,
    \begin{equation}
        \Tilde{\gamma}^{\infty}_{\cF}(\cN) \coloneqq \lim_{\epsilon \rightarrow 0} \limsup_{n\rightarrow \infty}\left(\gamma^{\epsilon}_{\cF_n}(\cN^{\ox n})\right)^{1/n},
    \end{equation}
    where $\cF_n$ is the free operation set of $n$-shot.
\end{definition}
With the definition of smoothed regularized $\gamma$-factor of $\cN$, in particular, we let $\cF = \LOCC$ this time and derive the following proposition connecting both the LOCC-knitting sampling cost of any bipartite quantum channels $\cN_{AB\rightarrow A'B'}$ and the (parallel) entanglement cost of the channel's Choi state $J^{\cN}_{AA'BB"}$, 

\begin{proposition}\label{app_prop:regularized_locc_lb_choi_state}
    For any bipartite channel $\cN_{AB\rightarrow A'B'}$, the regularized LOCC-assisted $\gamma$-factor of $\cN$ is lower bounded by,
    \begin{equation}
        \Tilde{\gamma}^{\infty}_{\LOCC}(\cN) \geq 2^{E_C(\Tilde{J}^{\cN}_{AA'BB'})}.
    \end{equation}
\end{proposition}

Before we talk about the proof of the proposition, we first investigate the situation for bipartite quantum states. Recalling the QPD of any bipartite quantum states regarding some convex free set $\cF$, considering the situation of virtually approximating a target state $\rho\in\cD(\cH_{AB})$. Suppose an error $\epsilon$-tolerance of the approximation for some fixed  $\epsilon>0$, we define the $\epsilon$-Ball of $\rho$ as,
\begin{equation}
    B_{\epsilon}(\rho) = \{\tau \in \cD(\cH_{AB})\mid d(\tau, \rho) \leq \epsilon\}.
\end{equation}
with respect to some distance metric $d(\cdot, \cdot):\cD(\cH_d)\times \cD(\cH_d) \rightarrow \RR$. In the following, we will count \textit{trace-norm} induced metric for the general discussion. We say $\Omega = \{(a_j, \sigma_j)\}_j$ forms an $\epsilon$-error QPD of $\rho$ regarding set $\cF$ if,
\begin{equation}
    \sigma = \sum_j a_j \sigma_j \in B_{\epsilon}(\rho); \quad a_j \in \RR; \quad \sigma_j \in \cF \ \forall j.
\end{equation}
If now we separate $\Omega$ into the positive and negative components and derive, 
\begin{equation}
    \sigma = \sum_{a_j \geq 0} a_j \sigma_j - \sum_{a_j < 0} |a_j| \sigma_j.
\end{equation}
Denoting $\kappa^{\pm} = \sum_{a_j \geq,< 0} |a_j|$, we have,
\begin{equation}
    \sigma = \kappa^+\sum_{a_j \geq 0} \frac{a_j}{\kappa^+}\sigma_j - \kappa^-\sum_{a_j < 0} \frac{|a_j|}{\kappa^-} \sigma_j.
\end{equation}
Particularly, since $\cF$ is now convex, we can define $\sigma_{1,2} = \sum_{a_j\geq, < 0} (|a_j|/\kappa^{\pm}) \sigma_j$, and $\sigma_{1,2}\in\cF$. Therefore, the decomposition protocol reduces two terms as $\{(c_1, \sigma_1),(c_2, \sigma_2)\}$. A similar discussion can also be found in Ref.~\cite{Harrow2024optimal}. To determine the minimum sampling cost, we run the following optimization program to range over all possible $\epsilon$-error decompositions as,
\begin{equation}
\begin{aligned}
    \gamma_{\cF}^{\epsilon}(\rho)&\coloneqq \min \{ c_1 + c_2 \;\sigma = c_1\sigma_1 - c_2 \sigma_2, \ c_{1,2} \geq 0, \ \sigma_{1,2}\in \cF, \sigma \in B_{\epsilon}(\rho) \}
\end{aligned}
\end{equation}
Once taking $\epsilon \rightarrow 0$, the above definition becomes the usual one-shot definition of $\gamma$-factor regarding the free set $\cF$ defined in~\cite{Piveteau2022quasiprobability}, i.e., $\gamma_{\cF}(\rho) = \lim_{\epsilon \rightarrow 0} \gamma_{\cF}^{\epsilon}(\rho)$. With this in hand, one can define the $\gamma$-factor of QPD obtaining zero error in the asymptotic regime,
\begin{definition}
    Given any quantum state $\rho$ and a free set $\cF$, the smoothed regularised $\gamma$-factor of $\rho$ with respect to $\cF$ is defined as,
    \begin{equation}
        \Tilde{\gamma}^{\infty}_{\cF}(\rho) \coloneqq \lim_{\epsilon \rightarrow 0} \limsup_{n\rightarrow \infty}\left(\gamma^{\epsilon}_{\cF_n}(\rho^{\ox n})\right)^{1/n},
    \end{equation}
    where $\cF_n$ is the free set of $n$-shot.
\end{definition}

\begin{lemma}\label{lem:smoothed_regularised_gamma_vs_regularised_gamma}
    For any quantum state $\rho$ and a free set $\cF$, $\Tilde{\gamma}^{\infty}_{\cF}(\rho) \leq \gamma^{\infty}_{\cF}(\rho)$.
\end{lemma}
\begin{proof}
    The proof follows by taking the optimal solution of $\gamma_{\cF_n}(\rho^{\ox n})$ for arbitrary $n\in \ZZ_+$, solution automatically forms the feasible solution for $\epsilon$-error $\gamma^{\epsilon}_{\cF_n}(\rho^{\ox n})$ for $\epsilon >0$. Therefore, by the continuity of $(\ \cdot\ )^{1/n}$ and taking the limitation, we prove,
    \begin{equation}
        \Tilde{\gamma}^{\infty}_{\cF}(\rho) =\lim_{\epsilon \rightarrow 0} \limsup_{n\rightarrow \infty}\left(\gamma^{\epsilon}_{\cF_n}(\rho^{\ox n})\right)^{1/n} \leq \limsup_{n\rightarrow \infty}\left(\gamma_{\cF_n}(\rho^{\ox n})\right)^{1/n} = \gamma^{\infty}_{\cF}(\rho).
    \end{equation}
\end{proof}

\begin{lemma}\label{lem:lb_by_smooth_rel_entropy}
Given $\rho \in \cD(\cH_{AB})$ and $\cF\subseteq\cD(\cH_{AB})$ is a convex free set, then the corresponding virtual sampling $\gamma$-factor is lower bounded,
\begin{equation}
    \gamma_{\cF}^{\epsilon}(\rho) \geq 2^{D_{\max}^{\epsilon}(\rho \| \cF)}
\end{equation}
\end{lemma}
\begin{proof}
    Recalling the definition of \textit{smoothed max-relative entropy} with respect to $\cF$. For density operators $\rho\in\cD(\cH_{AB})$, $\sigma\in \cF$, and a real number $\epsilon >0$, the $\epsilon$-\textit{smoothed relative max-entropy} of $\rho$ w.r.t $\cF$ is defined as,
\begin{equation}
\begin{aligned}         
    D_{\max}^{\epsilon}(\rho\|\cF) &= \min_{\tau\in B_{\epsilon}(\rho)} D_{\max}(\tau\|\cF)\\
    &= \min_{\tau\in B_{\epsilon}(\rho)} \inf_{\sigma\in\cF}\{\lambda\in \RR:\tau\leq 2^{\lambda} \sigma\} \\
    &= \min_{\tau\in B_{\epsilon}(\rho)} \inf_{\xi\in\cD(\cH_{AB})}\{\lambda\in \RR:\tau + (2^{\lambda} - 1)\xi = 2^{\lambda}\sigma, \sigma \in \cF \}
\end{aligned}
\end{equation}
By writing $c = 2^\lambda$ for some  $\lambda\in\RR$, we could then derive,
\begin{equation}
\begin{aligned}
    \gamma_{\cF}^{\epsilon}(\rho)&=\inf_{\lambda \in\RR} \{2\cdot 2^{\lambda}-1, \;\sigma + (2^{\lambda}-1)\sigma_2 = 2^{\lambda}\sigma_1, \sigma_{1,2} \in \cF, \sigma \in B_{\epsilon}(\rho)\}\\
    &\geq \min_{\sigma\in B_{\epsilon}(\rho)} \inf_{\sigma_{1} \in \cD(\cH_{AB})} \{2\cdot 2^{\lambda}-1,\;\lambda \in\RR, \sigma + (2^{\lambda}-1)\sigma_2 = 2^{\lambda}\sigma_1, \sigma_{2} \in \cF\}\\
    &\geq \min_{\sigma \in B_{\epsilon}(\rho)} 2\cdot 2^{D_{\max}(\sigma\|\cF)} - 1= 2\cdot 2^{D^{\epsilon}_{\max}(\rho\|\cF)} - 1.
\end{aligned}
\end{equation}
Now since $D_{\max}(\rho\|\sigma) \geq 0$ for any $\rho$ and $\sigma$, we have $2\cdot 2^{D_{\max}^{\epsilon}(\rho\|\cF)} - 1 \geq 2^{D_{\max}^{\epsilon}(\rho\|\cF)}$.
\end{proof}

\begin{proposition}\label{prop:regularised_EC_lb_states}
    Given a quantum state $\rho\in\cD(\cH_{AB})$, the smoothed regularized LOCC $\gamma$-factor of $\rho$ is lower bounded by,
    \begin{equation}
        \Tilde{\gamma}^{\infty}_{\LOCC}(\rho) \geq 2^{E_C(\rho)}.
    \end{equation}
\end{proposition}
\begin{proof}
    Fixing an integer $n\geq 1$, from Lemma~\ref{lem:lb_by_smooth_rel_entropy}, taking the separable states as the free set, i.e., $\cF = \SEP(A:B)$, we have,
    \begin{equation}
        \gamma_{\SEP_n}^{\epsilon}(\rho^{\ox n}) \geq 2^{D^{\epsilon}_{\max}(\rho^{\ox n}\|\SEP_n)}.
    \end{equation}
    Then we take the inner limit for the regularisation, which should preserve the inequality as it holds for any integer $n\geq 1$,
    \begin{equation}
        \limsup_{n\rightarrow \infty} (\gamma^{\e}_{\SEP_n}(\rho^{\ox n}))^{1/n} \geq \limsup_{n\rightarrow \infty}
        2^{\frac{1}{n}\rm{D}^{\epsilon}_{\max}(\rho^{\ox n}\|\SEP_n)} = 2^{\cE_{\max}^{\epsilon}(\rho)},
    \end{equation}
    where $\cE_{\max}^{\epsilon}(\rho)$ regarding some free state set $\cF$ is defined as,
    \begin{equation}
        \cE_{\max}^{\epsilon}(\rho) := \limsup_{n\rightarrow \infty} \frac{1}{n} \min_{\sigma \in \cF_n} D_{\max}^{\epsilon}(\rho^{\ox n} \|\sigma) = \limsup_{n\rightarrow \infty} \frac{1}{n} D_{\max}^{\epsilon}(\rho^{\ox n} \|\cF_n).
    \end{equation}
    As a result, taking the main results from~\cite{Datta2009max}, we then derive,
    \begin{equation}
        \Tilde{\gamma}^{\infty}_{\SEP}(\rho) = \lim_{\epsilon \rightarrow 0}\left(\limsup_{n\rightarrow \infty} (\gamma^{\epsilon}_{\SEP_n}(\rho^{\ox n}))^{1/n}\right) \geq \lim_{\epsilon \rightarrow 0} 2^{\cE_{\max}^{\e}(\rho)} = 2^{E_{R}^{\infty}(\rho)}.
    \end{equation}
    Then, since SEP channels contain all LOCCs and can never generate entanglement between two parties. Therefore, the  corresponding $E_{R}^{\infty}(\rho) = E_C(\rho)$ under the assistance of these channels~\cite{Brandao2010reversible}, we can then derive,
    \begin{equation}
        \Tilde{\gamma}^{\infty}_{\LOCC}(\rho) \geq 2^{E_C(\rho)},
    \end{equation}
    where the inequality holds because of the general relation between SEP channels and LOCCs.
\end{proof}

\vspace{3mm}
\begin{proof}
\textbf{of proposition~\ref{app_prop:regularized_locc_lb_choi_state}:} For any bipartite channel $\cN$ it holds that $\gamma^{\epsilon}_{\LOCC}(\cN)\geq \gamma^{\epsilon}_{\LOCC}(\Tilde{J}^{\cN}_{AA'BB'})$ as for the similar reason stated  in Ref.~\cite{Harrow2024optimal,Piveteau2022circuit}. This is not surprising by taking optimal QPD from $\gamma^{\epsilon}_{\LOCC}(\cN)$, one can construct the feasible solution of $\gamma^{\epsilon}_{\LOCC}(\Tilde{J}^{\cN}_{AA'BB'})$ by applying the QPD to the state $\Phi_{AA'}\ox \Phi_{BB'}$ and produce the state of $\epsilon$-faithfulness due to the diamond norm constraint. As a result, the inequality should hold for the $n$-copy scenario of $\cN$, which then leads to the proposition,
    \begin{equation}
        \Tilde{\gamma}^{\infty}_{\LOCC}(\cN) \geq 
        \Tilde{\gamma}^{\infty}_{\LOCC}(\Tilde{J}^{\cN}_{AA'BB'}).
    \end{equation}
\end{proof}

\begin{corollary}
    For any bipartite channel $\cN_{AB\rightarrow A'B'}$ that is bicovariant, then $\Tilde{\gamma}^{\infty}_{\LOCC}(\cN)$ is exponentially lower bounded by the (sequential) entanglement cost of $\cN$.
\end{corollary}
\begin{proof}
    The proof of the corollary is rather straightforward by applying the properties of bicovariant channel~\cite{Wilde2018,Bauml2019resource,Bauml2018}. Based on Proposition~\ref{app_prop:regularized_locc_lb_choi_state}, we automatically have,
    \begin{equation}
        \Tilde{\gamma}^{\infty}_{\LOCC}(\cN) \geq 2^{E_C(\Tilde{J}^{\cN}_{AA'BB'})} = 2^{E^{\circ}_C(\cN)},
    \end{equation}
    where we denote $E_{C}^{\circ}(\cN)$ as the sequential entanglement cost of $\cN$. In general, the sequential setting covers the parallel setting and therefore,
    \begin{equation}
        E_{C}^{\circ}(\cN) \leq E_{C}(\cN).
    \end{equation}
\end{proof}

\section{Circuit knitting for multiple noisy quantum gates}\label{appendix:circ_knit_distinct_channel}
In a more practical setting of circuit knitting, one has to consider the situation with noise. For example, considering the quantum circuit containing $n \geq 1$ distinct bipartite quantum channels or the noisy quantum gates, which can be seen as a reasonable practical assumption on NISQ devices, denoted as $\cN_1, \cdots, \cN_n$. To apply the circuit knitting, firstly, each of these $\cN_j$'s is aligned to parallel via local SWAP gates as shown in Fig.~\ref{fig:parallel_cut}.

Since local unitary would not change the PPT-assisted $\gamma$-factor, we denote the total channel $\cN = \bigotimes_{j=1}^n \cN_j$ and directly apply the single-copy lower bound Eq.~\eqref{eq:LN_max_nshot} to have,
\begin{equation}
    \gamma_{\PPT}\left(\cN\right)\geq \left(\prod_{j=1}^n 2^{LN_{\max}(\cN_j)}\right) - 1,
\end{equation}
where the equality holds due to the additivity of max logarithmic negativity, i.e., $LN_{\max}(\cN\ox\cM) = LN_{\max}(\cN) + LN_{\max}(\cM)$. As a consequence, by applying the relation between LOCC- and PPT-assisted $\gamma$-factor as well as the relation between max logarithmic negativity, exact PPT-entanglement cost, and bidirectional channel capacities, we have,
\begin{equation}
    \gamma_{\LOCC}(\cN) \geq \left(\prod_{j=1}^n 2^{E_{\PPT, C, 0}^{(1)}(\cN_j)}\right) - 1\geq  \left(\prod_{j=1}^n 2^{Q^{2\rightarrow 2}_{\LOCC}(\cN_j)}\right) - 1.
\end{equation}

However, in the practical scenario, one has to be aware that cutting directly on the total circuit as determining the entire process of $\cN$ is typically intractable. Instead, for each of these $n$ channels, suppose the corresponding optimal LOCC-assisted QPD has been pre-determined, and denote the associated optimal sampling overheads by $\gamma_{\LOCC}(\cN_{j})$ for the $j$-th channel. One has to treat each channel $\cN_j$ independently and randomly replace them with one of its QPD channels. This leads to a total number of samples by at least $\gamma_{\rm tot}^2 = \prod_{j=1}^n \gamma_{\LOCC}(\cN_j)^2$. We can have the following corollary.
\begin{corollary}
    For a quantum circuit containing $n$ bipartite quantum channels $\cN_1, \cN_2, \cdots, \cN_n$, the practical sampling overhead is lower bounded by,
    \begin{equation}
        \gamma_{\rm tot}^2 \geq \prod_{j=1}^n \left(2^{ E^{(1)}_{\cF,C,0}(\cN_j)} - 1\right)^2.
    \end{equation}
    where $\cF \in \{\SEP, \PPT\}$.
\end{corollary}
Taking the definition of $\gamma_{\rm tot}$, we have,
\begin{equation}
\begin{aligned}
    \gamma_{\rm tot}^2 &= \prod_{j=1}^n \gamma_{\LOCC}(\cN_j)^2 \geq \prod_{j=1}^n \gamma_{\cF}(\cN_j)^2 \geq \prod_{j=1}^n \left(2^{ E^{(1)}_{\cF,C,0}(\cN_j)} - 1\right)^2.
\end{aligned}
\end{equation}
The last inequality holds due to the Lemma derived in the previous sections. Given these nonlocal channels, by assumption, not all of them own zero one-shot exact entanglement cost with respect to $\cF$, and we denote the maximal one-shot exact cost among $\cN_j$s as $\widehat{E}_{\cF,C}$. Then the above bound reduces to the situation of copies of a fixed channel as
\begin{equation}
    \gamma_{\rm tot}^2 \geq  2^{2\widehat{E}_{\cF,C}} - 2^{\widehat{E}_{\cF,C} + 1} + 1.
\end{equation}
Therefore, in the large dimension limit, the total number of samples required scales as $\cO(4^{\widehat{E}_{\cF,C}})$. Assuming the two parties of the system are isomorphic with dimension $d_A = d_B = d$, particularly for a qubit system, $d = 2^N$ where $N$ is the number of qubits in the half circuit. The one-shot exact entanglement cost of the bipartite channels could grow in $\log(d)$, which then leads to the scaling of $\gamma_{\rm tot}^2$ as $\cO(\exp(N))$. Such an exponential growth in the sampling cost leads to the fundamental limitation of practically realizing the circuit knitting method applies to $n$ nonlocal noisy gates on NISQ devices, which has been also pointed out in~\cite{Piveteau2022circuit}.

\end{document}